\documentclass[11pt,onecolumn,draftcls]{IEEEtran}

\usepackage{adjustbox}
\usepackage{tikz}
\usetikzlibrary{quantikz}
\tikzstyle{tensor}=[rectangle,draw=blue!50,fill=blue!20,thick]
\tikzstyle{input}=[circle,draw=blue!50,fill=blue!20,thick]
\usepackage[utf8]{inputenc} 
\usepackage[T1]{fontenc}    
\usepackage{hyperref}       
\usepackage{url}            
\usepackage{booktabs}       
\usepackage{amsfonts}       
\usepackage{nicefrac}       
\usepackage{microtype}      
\usepackage[normalem]{ulem}
\usepackage{subcaption}
\usepackage{bm}
\usepackage{braket}
\usepackage{graphicx}
\usepackage{amsthm}
\usepackage{algorithm}
\usepackage{algpseudocode}
\newtheorem{thm}{Theorem}
\newtheorem{lem}[thm]{Lemma}
\newtheorem{definition}{Definition}
\newtheorem{prop}[thm]{Proposition}
\newtheorem{coro}[thm]{Corollary}
\newtheorem{clam}[thm]{Claim}
\usepackage{mathtools}
\usepackage{lipsum}
\usepackage{nccmath}
\usepackage{xcolor}
\usepackage{amssymb}
\usepackage{caption}
\usepackage{float}
\font\myfont=cmr12 at 20pt
\DeclareMathOperator{\Tr}{Tr}

\begin{document}
	
	\title{\myfont Efficient Online Quantum Generative Adversarial Learning Algorithms with Applications}

	
\author{
	\IEEEauthorblockN{Yuxuan Du\IEEEauthorrefmark{1}, Min-Hsiu Hsieh\IEEEauthorrefmark{2}, Dacheng Tao\IEEEauthorrefmark{1}}
	
	\IEEEauthorblockA{\IEEEauthorrefmark{1}UBTECH Sydney Artificial Intelligence Centre and the School of Information Technologies, Faculty of Engineering and Information Technologies, The University of Sydney, Australia}
	
	\IEEEauthorblockA{\IEEEauthorrefmark{2}Centre for Quantum Software and Information, Faculty of Engineering and Information Technology, University of Technology Sydney, Australia
		}
}
	
	
	
	\date{\today}

	
	\maketitle

	\begin{abstract}
The exploration of quantum algorithms that possess quantum advantages is a central topic in quantum computation and quantum information processing. One potential candidate in this area   is quantum generative adversarial learning (QuGAL), which conceptually has  exponential advantages over classical adversarial networks. However, the corresponding learning algorithm remains obscured. In this paper, we propose the first quantum generative adversarial learning algorithm\textemdash the quantum multiplicative matrix weight algorithm (QMMW)\textemdash which enables the efficient processing of fundamental tasks. The computational complexity of QMMW is polynomially proportional to the number of training rounds and  logarithmically proportional to the input size. The core concept of the proposed algorithm combines QuGAL with online learning. We exploit  the  implementation of QuGAL with parameterized quantum circuits, and numerical experiments for the task of entanglement test for pure state are provided to support our claims.
\end{abstract}

\section{Introduction}
The principal interest in quantum computation is the exploration of potential applications that outperform their classical counterparts. The rapid development of  quantum hardware divides this interest into short-term and long-term goals. The short-term goal is to devise quantum algorithms that not only possess quantum advantages but can also be implemented on near-term devices \cite{preskill2018quantum}. The long-term goal is to employ fault-tolerant quantum computers that are capable of providing remarkable quantum speedups over classical methods \cite{shor1999polynomial} to tackle practical real-world problems.    
 
Quantum machine learning is one of the most promising candidates for achieving both  short-term and long-term goals   \cite{biamonte2017quantum}, and the proposed quantum  generative adversarial learning (QuGAL)  strengthens this belief \cite{lloyd2018quantum}. The main theoretical conclusion of QuGAL is that  exponential quantum advantages may exist under the assumption that the target data distribution can be efficiently encoded into a density matrix \cite{lloyd2018quantum}. Conceptually,  QuGAL involves two players, a generator and a discriminator, which play a zero-sum game. At each training round, the generator tries to approximate the target data to fool the discriminator, while the discriminator tries to distinguish the fake data from the real data. When the generator and discriminator are both constructed by  quantum operations, the adversarial quantum learning game has the potential to converge to Nash equilibrium with an exponential speedup.
  
Despite promising theoretical results, two issues related to QuGAL have not been explored. First, it is unclear what kinds of learning tasks can be accomplished by QuGAL to potential advantages. Second, an explicit learning algorithm of QuGAL that can fast converge to the equilibrium remains unexplored.  Previous studies mainly focus on the implementation of QuGAL under near-term quantum devices as so-called quantum generative adversarial networks (QuGANs)  \cite{dallaire2018quantum,situ2018adversarial,zeng2018learning,romero2019variational,zoufal2019quantum}. In particular, the generator and discriminator of QuGANs are constructed by employing  parameterized quantum circuits (PQCs) that  are composed of a set of trainable parameterized single qubit gates and two-qubit CNOT gates   \cite{du2018expressive}.  However, the intrinsic optimization mechanism of PQCs that iteratively updates each gate  destroys the required convex-concave property in QuGAL, which implies that the obtained result may not converge to Nash equilibrium and may induce additional training difficulties, e.g.,  mode collapse and vanishing gradients \cite{arjovsky2017wasserstein}. Two key issues therefore exist for QuGANs, i.e.,  how to improve stability and convergence in training QuGANs, and whether QuGANs deliver potential quantum advantage.
 
 To tackle the aforementioned issues, we revisit the theory of QuGAL in this paper from the perspective of  online learning \cite{hazan2016introduction}. The integration of online learning with QuGAL is motivated by the fact that online learning algorithms can efficiently approximate the optimal result for the zero-sum game associated with the convex-concave property, and the training of QuGAL satisfies this condition. This  observation enables us to devise a quantum adversarial learning algorithm with online learning features and to theoretically analyze its potential quantum advantages.  Additionally, online learning has been employed as a powerful tool for  relieving training difficulties in classical generative adversarial networks  (GANs) \cite{grnarova2018an}, which motivates us to introduce such a method in optimizing QuGANs. Lastly, we  investigate how to use QuGAL to accomplish  learning quantum information processing tasks, such as the quantum entanglement test for pure state and quantum state discrimination \cite{horodecki2009quantum,chitambar2014asymptotic,chitambar2014local,chitambar2013revisiting,chitambar2017round}. Our study opens avenues for exploring  quantum information processing tasks using  quantum generative  adversarial learning models.

We summarize the  main results of this work as follows.
 \begin{itemize}
 \item We propose a quantum generative adversarial learning algorithm, the quantum multiplicative matrix weight (QMMW) algorithm, which rapidly  converges to Nash equilibrium as  expected from QuGAL. QMMW is inspired by the multiplicative matrix weight algorithm, which is a popular online learning algorithm that efficiently finds optimal solutions to the zero-sum game \cite{kale2007efficient}. We prove that  the convergence  rate of QMMW is $\mathcal{O}(\sqrt{N/T})$, where $N$ is the number of qubits corresponding to the target density matrix and $T$ is the number of training rounds. An attractive feature of QMMW is that the output states of both the generator and discriminator  can be viewed as Gibbs states. By exploiting the efficient Gibbs sampling method proposed by \cite{van2018improvements}, we prove   that the computational complexity of QMMW   is  $\mathcal{O}(N^3T^4)$. 

	\item We introduce a multiplicative weight training method to overcome the training difficulty encountered in QuGANs. The core ingredient of this  method is to seek the most possible optimized direction for achieving global equilibrium through the inherent mechanism of online learning. In the training process, a multiplicative weight training method  puts more weight to the gradient that is more probable to fool the discriminator.  Since the multiplicative weight training method only focuses on re-weighting the gradient, it can be seamlessly embedded into other optimization methods used in QuGANs.
	\item We investigate the potential quantum advantages by applying  QMMW and QuGANs to solve quantum information tasks,  i.e., the pure state entanglement test and the quantum state discrimination. In particular, we numerically validate that QuGANs are  capable of accomplishing the pure state entanglement test with modest quantum resources, which sheds light on using QuGANs to handle other quantum information learning tasks. All numerical simulations demonstrated in this paper are implemented in Python, leveraging the pyQuil and QuTiP libraries to access the numerical simulators \cite{smith2016practical,johansson2012qutip}.
 \end{itemize}

\subsection{Related Works}
Online convex optimization  has been broadly applied to the study of linear programming, semidefinite programming, and zero-sum game \cite{kale2007efficient,hazan2016introduction}. Recently, it has been   employed to study shadow quantum tomography \cite{aaronson2018online}. An advanced meta-algorithm in online convex learning, the so-called multiplicative weight, has been introduced to study the quantum zero-sum game algorithm \cite{van2019quantum},  non-interactive zero-sum quantum games \cite{jain2009parallel},  the parallel approximation of semidefinite programs and minmax problems \cite{gutoski2013parallel},  and quantum semi-definite programming \cite{2017arXiv171002581B}. Despite their similarities, the various studies, including this work, have adopted  different update rules and focused on different tasks, leading to distinct  theoretical results on, for example,  convergence rate. 

In the rest of this paper,  we first introduce QuGAL and discuss its applications on quantum information processing in Section \ref{Sec:prob_app}. In  Section \ref{Sec:QMMW}, we describe QMMW and theoretically analyze its computational  cost. In  Section \ref{Sec:QuGAN}, we give the multiplicative weight training method for QuGANs. In Section \ref{Sec:QMM_Qu_app}, we explain how to employ QMMW and QuGANs to tackle quantum information processing tasks.  In Section \ref{Sec:Simu}, we numerically validate the effectiveness of applying QuGANs to accomplish quantum information processing  tasks. Section \ref{Sec:concl} concludes the paper.

\section{Quantum generative adversarial learning and its applications}\label{Sec:prob_app}
We formally define the quantum generative adversarial learning (QuGAL) problem  and devise a general  framework for using  QuGAL to accomplish quantum information processing tasks.  Suppose that a given mixed state $\rho$ is represented by $N$ qubits,  the goal of QuGAL is  to reproduce $\rho$. QuGAL employs two players to set up a zero-sum game  \cite{osborne1994course}: The first player refers to a generator, which generates a mixed state $\sigma$ to approximate $\rho$; the second player refers to the discriminator $\mathcal{D}$, which aims to  maximally distinguish $\rho$ from $\sigma$. Such a zero-sum game is evaluated by a loss function $\mathcal{L}(\cdot,\cdot)$, where its physical meaning is the classification error. In the training process, the generator tries to minimize the loss function while the discriminator tries to maximize it. By labeling the state $\rho$ as `True' and the generated state  $\sigma$ as `False', we have  
\begin{equation}\label{eqn:3}
\mathcal{L}(\sigma, \mathcal{D}) = P(\text{True}|\sigma)P(G) + P(\text{False}|\rho)P(R)~,
\end{equation}
where $P(G)$ ($P(R)$) refers to the prior of operating the discriminator  with $\sigma$ ($\rho$), and $P(\text{True}|\sigma)$ ($P(\text{False}|\rho)$ refers to the likelihood that the discriminator   will classify $\sigma$ ($\rho$) as `True' (`False'). Throughout this paper, we set $P(G)=P(R)=1/2$. 

When the discriminator is  assigned to be  positive operator value  measurement (POVM), we have $P(\text{True}|\sigma)=\Tr((\mathbb{I}-\mathcal{D})\sigma)$ and   $P(\text{False}|\rho)=\Tr(\mathcal{D}\rho)$, where the corresponding loss function possesses the convex-concave property. This property immediately indicates that equilibrium always exists,  guaranteed by the theoretical result of the convex optimization \cite{boyd2004convex}. Denoting the optimal solution as $(\sigma^*, \mathcal{D}^*)$ with $(\sigma^*, \mathcal{D}^*) = \arg\max_{\sigma}\min_{\mathcal{D}}\mathcal{L}(\sigma, \mathcal{D})$, we have $\sigma^*=\rho$ with $\mathcal{L}(\sigma^*, \mathcal{D}^*)= 1/2$ at the equilibrium point. 
 
Despite this promising property, the means of applying QuGAL to solve certain problems with  an exponential quantum advantage is unknown. In the following, we propose a general principle for employing QuGAL to solve quantum information processing problems \cite{nielsen2010quantum}.
Let us recall a common strategy  of conventional methods in quantum information processing tasks, e.g., quantum entanglement test or  quantum state discrimination \cite{barnett2009quantum,horodecki2009quantum}. 
A conventional method generally has two steps: Extraction of quantum information into the classical forms, followed by  manipulation of the collected classical data into the desired result. Due to the curse of dimensionality, the number of measurements required to collect a sufficient amount of quantum information grows exponentially with respect to the number of qubits. In contrast, the desired result is often unrelated  to the size of the input and can be represented in a low dimensional space. For instance, the outcome of an  entanglement test 
is binary, indicating whether the input state is entangled or not. 
Applying QuGAL to manipulate quantum data and output the result directly could immediately have exponential quantum advantage, assured by circumventing the enormous amount of quantum measurements required by  conventional methods. 

The key issue, given this observation, is identifying how to reformulate a given quantum information processing task as quantum generative adversarial learning language. Here we  devise a general framework to achieve this goal.  {The central idea behind this framework is  to  conditionally limit the expressive power of the generator or the discriminator for a given task. Adopting the constraint operation aims to distinguish the desired answer from other results, where the training  loss of QuGAL will conditionally converge to the Nash equilibrium if and only if the given input directly relates to the desired answer. In other words, the convex-concave property of the training loss defined in Eqn.~(\ref{eqn:3}) is conditionally preserved when the input directly relates to the desired answer.  We outline the framework as follows. First, we translate a quantum information processing task into a binary decision problem, which can be effectively  achieved by employing the `one-versus-all' strategy \cite{witten2016data}. We then constrain the expressive power of the generator or the discriminator, e.g., the generator can only well approximate all possible inputs corresponding to the desired answer.  If the discrepancy between  the obtained loss and  the optimal loss is below a certain threshold after training,, the desired outcome is obtained.} The restriction of the expressive power particularly depends on the detailed setting and implementation of QuGAL, and we will illustrate how to limit the expressive power in the following sections.  

We illustrate how to tackle entanglement test problems under the proposed QuGAL framework. {Let us briefly review the entanglement test. The entanglement test targets to the detection of  whether a given quantum state is entangled or separable.  
Devising an efficient separability criteria to distinguish entanglement for specific quantum states is fundamentally important for quantum applications. Previous separability criteria can be roughly divided into two classes \cite{horodecki2009quantum}: (1) The separability criteria are efficient but incomplete, that is, some entangled states could be misclassified as separable states, e.g., positive partial transposition  \cite{horodecki1997separability}. (2) The separability criteria are complete in the sense that they are capable of correctly identifying any entangled states at any dimensions, but the computational cost is very high, e.g., symmetric  extension \cite{doherty2004complete}. Besides these two conventional classes, machine learning, which has been explored as an effective tool in many physics problems \cite{carrasquilla2017machine,van2017learning}, also provides novel insights into the tasks of the entanglement test \cite{lu2018separability,ma2018transforming}. A common weakness of the above methods is the requirement for quantum state tomography to construct the classical density matrix, which leads to an  exponential runtime with a linearly increased  number of qubits.}

{QuGAL is a potential candidate for overcoming  the aforementioned issue by directly manipulating the quantum data to circumvent the time-consuming quantum state tomography. The separability rule of QuGAL is reflected by the training loss such that  the given state is classified as being  entangled if the training loss cannot converge to the equilibrium below a threshold after a certain number of training rounds.  The detailed procedure of  employing QuGAL to tackle the  entanglement test is as follows.} Given an unknown quantum state $\rho$, the entanglement test asks if $\rho$ is entangled or not, which is a binary decision problem.
In this setting, we restrict the expressive power of the generator to only generate separable states. If $\rho$ is separable, the convex-concave property of QuGAL is preserved, where $\sigma$ output by the generator can efficiently approximate $\rho$ and the loss converges to the Nash equilibrium very quickly. Otherwise, the convex-concave property is lost and the Nash equilibrium can never be reached.  The fact that the loss of QuGAL can be efficiently calculated by  two outcome measurements immediately gives QuGAL  an exponential quantum advantage over conventional methods, which require  exponential measurements with respect to the number of qubits.

\section{Quantum Multiplicative Matrix Weight}\label{Sec:QMMW}
The convex-concave property of QuGAL enables us to ues the results of  convex optimization,  under the no-regret framework for online learning \cite{hazan2016introduction}, to develop an advanced quantum  algorithm that is capable of fast convergence to the equilibrium.   We first give the definition of $\textit{regret}$ before moving on to explain how no-regret learning algorithms work. Given a sequence of convex loss functions $\{\mathcal{F}_1,\mathcal{F}_2,...,\mathcal{F}_T\}: \mathcal{K}\rightarrow \mathbb{R}$, an algorithm $\mathcal{A}$ selects a sequence of $K_t\in\mathcal{K}$'s with $\mathcal{K}$ being the input space, each of which may only depend on previously observed $\{\mathcal{F}_1,...,\mathcal{F}_{t-1}\}$. The algorithm $\mathcal{A}$ is said to have \textit{no regret} if its minimized regret  $\min_K R_T(K) = \mathcal{O}(T)$, where we define $R_T(K):=\sum_{t=1}^T(\mathcal{F}_t(K_t) - \min_{K\in \mathcal{K}}\mathcal{F}_t(K))$. 

Here we propose a no-regret quantum generative adversarial learning algorithm\textemdash the quantum multiplicative matrix weight (QMMW) algorithm\textemdash to efficiently reconstruct the given mixed state under the fault-tolerant quantum circuits setting.  Conceptually, QMMW is inspired by the multiplicative matrix weight algorithm \cite{kale2007efficient}, an advanced meta-algorithm with the no-regret property that is broadly used in online convex optimization  \cite{hazan2016introduction}.  

Before presenting the technical treatment,  we explicitly define the generated state,  the discriminator,  and  the loss function in Eqn.~(\ref{eqn:3}) used in  QMMW. We denote the output states of the generator and the discriminator as $\sigma_G$ and $\sigma_{D}$, respectively.  The loss function at $t$-th round  $\mathcal{L}(\sigma_G^{(t)}, \sigma_D^{(t)})$ is 
\begin{equation}\label{eqn:QMMW_loss1}
\mathcal{L}(\sigma_G^{(t)}, \sigma_D^{(t)})= \frac{1}{2}\left(\Tr\left(\sigma_D^{(t)}\rho\right) - \Tr\left(\sigma_D^{(t)}\sigma_G^{(t)}\right) \right)+\frac{1}{2} 
~.
\end{equation} 
The physical meaning of this loss function is the evaluation of the overlap  between $\rho$ and $\sigma_G$  using $\sigma_D^{(t)}$ \footnote{As discussed in Section 3, the mixed state will be purified in the implementation of QMMW, where the physical meaning of the loss will be clearer.}.  The convexity of trace calculation implies   that the loss function defined in  Eqn.~(\ref{eqn:QMMW_loss1}) has the convex-concave property, where the optimal solution is $\sigma_G^*=\rho$ with the equilibrium  $\mathcal{L}(\sigma_G^{*}, \sigma_D^{*})=1/2$. 

Following the theoretical results of approximating Nash equilibrium, an  algorithm that has the no-regret property will quickly converge to the equilibrium \cite{farina2017regret}. We denote that the regret for the generator and discriminator as $R_T({\sigma_G})$ and $R_T{(\sigma_D)}$, respectively. We will prove later that  QMMW possesses the  no-regret property with $R_T(\sigma_G)=\mathcal{O}(T)$ and $R_T(\sigma_D) =\mathcal{O}(T)$, which implies that  the optimal result can be efficiently located.   Mathematically, the minimized regret for the generated state $\sigma_G$ during $T$ training rounds  is defined as  
$$R_T(\sigma_G) =-\sum_{t=1}^T\mathcal{L}(\sigma_G^{(t)},\sigma_D^{(t)})+\min_{\sigma_G} \sum_{t=1}^T\mathcal{L}(\sigma_G,\sigma_D^{(t)})~.$$  
Similarly, we can define the regret for the discriminator as  
$$R_T(\sigma_D)=  \sum_{t=1}^T\mathcal{L}(\sigma_G^{(t)},\sigma_D^{(t)})-\min_{\sigma_D} \sum_{t=1}^T\mathcal{L}(\sigma_G^{(t)},\sigma_D)~.$$  

We now explain QMMW. QMMW consists of three steps. First, given a targeted state $\rho$ represented by $N$ qubits, we set the total number of training rounds as $T$ and let the tolerable error  be $\epsilon=\sqrt{N/T}$ with $\epsilon \leq 1/2$. We also initialize the  discriminator $\sigma_D^{(1)}$ as the maximally  mixed state $\sigma_D^{(1)}=\mathbb{I}/2^{N}$. Second, we iteratively update the generator and the discriminator $T$ training rounds.  The update rule for the generated state $\sigma^{(t)}$ at $t$-th round is 
\begin{equation}\label{eqn:upd_G}
\sigma_G^{(t)}=\frac{\exp^{\sum_{\tau=1}^t(-\epsilon \sigma_D^{(\tau)})} }{\Tr\left(\exp^{\sum_{\tau=1}^t(-\epsilon \sigma_D^{(\tau)})}\right)}~.
\end{equation} 
The update rule for the discriminator is
\begin{equation}
	\sigma_D^{(t+1)} =  \frac{e^{\sum_{\tau =1}^t -\epsilon \left( \rho  - \sigma_G^{(\tau)} \right)}}{\Tr{e^{\sum_{\tau =1}^t -\epsilon \left( \rho  - \sigma_G^{(\tau)} \right)}}}~.
\end{equation}

 Third, we calculate the loss  $\mathcal{L}(\bar{\sigma}_D, \bar{\sigma}_D)$ defined in  Eqn.~(\ref{eqn:QMMW_loss1}) with  $\bar{\sigma}_G=\sum_{t=1}^T{\sigma_G^{(t)}}/{T}$  and the  averaged discriminator $\bar{\sigma}_D=\sum_{t=1}^T{\sigma_D}^{(t)}/T$ during $T$ training rounds.
This convergence  rate of QMMW is assured  by the following theorem:
\begin{thm}\label{thm2}
Given a mixed state $\rho$ represented by $N$ qubits, and setting the training rounds as $T$, QMMW  yields 
\begin{equation}
\left|\mathcal{L}(\bar{\sigma}_G, \bar{\sigma}_D) - \mathcal{L}({\sigma_G^*}, {\sigma_D^*})\right|
\leq 3\sqrt{\frac{N}{T}}.
\end{equation}
 \end{thm}
The proof of Theorem \ref{thm2} is given in the supplementary material SM (A).  

QMMW can be efficiently executed on fault-tolerant  quantum circuits, since both $\sigma_G^{(t)}$ and $\sigma_D^{(t)}$ are Gibbs states that can be prepared by using efficient Gibbs sampling methods \cite{2017arXiv171002581B,van2018improvements}.  We elaborate how to carry out the proposed QMMW algorithm in the supplementary material SM(B). The efficiency of the Gibbs sampling methods proposed in \cite{van2018improvements} presents another attractive advantage of QMMW:
\begin{thm}\label{thm:exp}
 Given an $N$-qubit state, let $U_{\rho}$ be the unitary that prepares the purification state of $\rho$. Denote $T$ as the total  number of training rounds. If there is  quantum query access to $U_{\rho}$, the computation cost of  the QMMW algorithm is   $\mathcal{O}(N^{3}T^4)$.
\end{thm}
The proof of Theorem \ref{thm:exp} is given in the supplementary material SM (C). We remark that allowing access the Gibbs sampler $U_{\rho}$ has also been used in \cite{2017arXiv171002581B}.

\section{QuGANs with multiplicative weight training method}\label{Sec:QuGAN}
{The investigation of applying QuGANs to tackle quantum information processing problems is of practical interest in the near term when there are only limited available qubits and shallow quantum circuit depth \cite{preskill2018quantum}.  } Although several studies have confirmed the feasibility of using QuGANs to achieve  certain tasks, the variational optimization  method collapses the desired convex-concave property and heavily challenges the performance of QuGANs. {The disappearance of the convex-concave property results in an inevitable difficulty, since  the optimization may get  stuck in local minima.}   This topic has been widely investigated in classical GANs \cite{zhang2016understanding}. Inspired by the weighted training algorithm proposed by \cite{pantazis2018training},  which has demonstrated its effectiveness in classical GANs, we propose the multiplicative weight training method \cite{pantazis2018training} to relieve the training difficulty in QuGANs. The proposed training method can be  seamlessly embedded into  advanced optimization   algorithms used to train  parameterized quantum circuits (PQCs).

Before illustrating how the multiplicative weight training method works, we first set up the QuGAN used in this paper. The generator $U_G$ and discriminator $U_D$ of our  QuGAN are two trainable unitaries that are implemented by PQCs. Mathematically,  the trainable unitary $U_G$ and $U_D$ are  defined as 
\begin{equation}\label{eqn:PQC_U_G}
U_G=\prod_{i=1}^{L_1}U(\bm{\theta_i}), ~ U_D=\prod_{i=1}^{L_2}U(\bm{\gamma_i})	~,
\end{equation}
where $L_1$ ($L_2$) refers to the number of blocks in $G$ ($D$) and each block $U(\bm{\theta_i})$ ($U(\bm{\gamma}_i)$) has an identical arrangement of quantum gates.  Suppose that the target state $\rho$ is represented by $N$ qubits,  the generated state is formulated as $\ket{\phi}=U_G\ket{0}^{\otimes N'}$ with $N'=N+N_a$ and $N_a$ being the number of ancillary qubits \footnote{If the given state is a pure state, we have $N_{a}=0$. The value of $N_a$ is no larger than $N$.}. The generated mixed state $\sigma_G$ can be obtained by partial tracing   the ancillary system, i.e., 
\begin{equation}\label{eqn:Q_GAN_6}
	\sigma_G= \Tr_a(U_G(\ket{0}\bra{0})^{\otimes N'}U_G^{\dagger}),
\end{equation}
supported by Stinespring’s dilation theorem \cite{nielsen2010quantum}. The discriminator of our QuGANs is defined as 
\begin{equation}\label{eqn:QGAN_U_D}
M_D  = U_D^{\dagger}
(\mathbb{I}\otimes E_F)U_D~,
\end{equation}
where a two-outcome positive-operator valued measurement $\mathcal{D}$ defined in Eqn.~(\ref{eqn:3})  is  reformulated as $U_D$  followed by a partial measurement  $E_F=\ket{0}\bra{0}$ on an ancillary qubit. 
Following the loss function of QuGAL defined in Eqn.~(\ref{eqn:3}), the loss function of QuGAN yields 	
\begin{equation}\label{eqn:5} 
	\mathcal{L}(U_G,U_D)=
	 \Tr((M_D(\rho\otimes\ket{0}\bra{0})) P(R) + \Tr((\mathbb{I}-M_D)(\sigma_G\otimes \ket{0}\bra{0}))P(G)~,
	\end{equation}
	where $\sigma_G$ is defined in Eqn.~(\ref{eqn:Q_GAN_6}) and $M_D$ is defined in Eqn.~(\ref{eqn:QGAN_U_D}). The loss function of QuGAN gives the following theorem:
\begin{lem}\label{thm1}
The loss function  $ \mathcal{L}(U_G,U_D)$ defined in Eqn.~(\ref{eqn:5}) has the convex-concave property with the equilibrium value $\mathcal{L}(U_G^*,U_D^*)=1/2$.
\end{lem}
The proof of Theorem \ref{thm1} is given in the supplemental material SM (D). 


We now illustrate how to use the  multiplicative weight training method to facilitate the optimization of QuGAN. Intuitively, this method aims to put more weight on generated states that are more likely to fool the discriminator in updating $U_G$.  We summarize the multiplicative weight training method  in Algorithm \ref{alg:1}. 

\begin{algorithm}
\caption{The Multiplicative Weight Training Method}
\label{alg:1}
\hspace*{\algorithmicindent} \textbf{Input}:  $T$; $K$;  $\alpha\in \mathbb{R}$; $\eta\in(0,1)$. \\
\hspace*{\algorithmicindent} \textbf{Output}: The trainable parameters $\bm{\theta}^{(T)}$ and $\bm{\gamma}^{(T)}$ .
\begin{algorithmic}[1]
\State Initialize trainable parameters  $\bm{\theta}^{(0)}$, $\bm{\gamma}^{(0)}$;  \Comment{Randomly sampled from uniform distribution.}
 \For{\texttt{$t=1;$ $t\leq T;$ $t\leftarrow t+1$}}
        \State  $\tilde{\bm{\theta}}^{(1)}\leftarrow \bm{\theta}^{(t)}$ and $\tilde{\bm{\gamma}}^{(1)}\leftarrow \bm{\gamma}^{(t)}$ \Comment{Initialize $\tilde{\bm{\theta}}^{(k)}$ and  $\tilde{\bm{\gamma}}^{(k)}$ for $k=1$}
        	 \For{\texttt{$k=1;$ $k\leq K;$ $k\leftarrow k+1$}}
        \State  $\{\mathcal{L}(U_G^{(k)},U_D^{(k)}),  \nabla_{\tilde{\bm{\theta}}} \mathcal{L}(U_G^{(k)},U_D^{(k)})\}$   \Comment{Record the training loss and gradients}
        \State $\tilde{\bm{\theta}}^{(k+1)}\leftarrow\tilde{\bm{\theta}}^{(k)} + \alpha \nabla_{\tilde{\bm{\theta}}} \mathcal{L}(U_G^{(k)},U_D^{(k)})$    \Comment{Update the virtual parameters $\tilde{\bm{\theta}}$ }
        \State $\tilde{\bm{\gamma}}^{(k+1)}\leftarrow\tilde{\bm{\gamma}}^{(k)} - \alpha \nabla_{\tilde{\bm{\gamma}}} \mathcal{L}(U_G^{(k)},U_D^{(k)})$ \Comment{Update the virtual parameters $\tilde{\bm{\gamma}}$ }
  \EndFor
  \State $w_k=\eta\frac{{\mathcal{L}(U_G^{(k)},U_D^{(k)})}}{\sum_{k=1}^K{\mathcal{L}(U_G^{(k)},U_D^{(k)})}}$ \Comment{Calculate the multiplicative weights  $\{w_k\}_{k=1}^K$}
  \State $\bm{\theta}^{(t+1)}\leftarrow \bm{\theta}^{(t)} + \alpha \sum_{k=1}^Kw_k \nabla_{\bm{\theta}} \mathcal{L}(U_G^{(k)},U_D^{(k)})$ \Comment{ Update the trainable parameters for $U_G$}
   \State $\bm{\gamma}^{(t+1)}\leftarrow \bm{\gamma}^{(t)} -  \alpha \nabla_{\bm{\gamma}} \mathcal{L}(U_G^{(t)},U_D^{(t)})$ \Comment{Update the trainable parameters for $U_D$}
  \EndFor
\end{algorithmic}
\end{algorithm}

The four hyper-parameters of the multiplicative weight training method  are the total number of training rounds $T$, the total number of inner iterations  $K$,  the learning rate $\alpha\in \mathbb{R}$,  and the scale parameter $\eta\in(0,1)$. At each training round $t$ with $t\in[T]$, we introduce $K$ inner iterations to obtain a better  gradient for updating $\bm{\theta}^{(t+1)}$.  For ease of understanding, we denote the updated parameters in $K$ iterations as $\tilde{\bm{\theta}}$ and $\tilde{\bm{\gamma}}$. As indicated in  Lines $5$-$7$ of Algorithm \ref{alg:1}, we iteratively update $\tilde{\bm{\theta}}$ and $\tilde{\bm{\gamma}}$, and  record a set of training losses $\{\mathcal{L}(U_G^{(k)},U_D^{(k)})\}_{k=1}^K$ and a set of gradients $\{\nabla_{\bm{\theta}} \mathcal{L}(U_G^{(k)},U_D^{(k)})\}_{k=1}^K$.  After conducting the inner iterations, we calculate the multiplicative weights and  employ them to update $\bm{\theta}$, as indicated by Lines $9$-$10$ in Algorithm \ref{alg:1}. We note that the multiplicative weight training method differs from the weighted training algorithm proposed in \cite{pantazis2018training}. The major difference is in the mechanism of QuGAN and classical GANs, i.e.,  classical GANs support nonlinear mapping, whereas QuGAN can only conduct linear mapping (see more details about classical  GANs in the supplementary material SM (E)).

\section{The application of QMMW and QuGANs for entanglement test}\label{Sec:QMM_Qu_app}
Following the observation in Section \ref{Sec:prob_app}, a  core ingredient of employing QuGAL to tackle a given quantum information processing problem is to conditionally restrict the expressive power of the generator or discriminator. The restriction method is varied for different settings and  implementations of QuGAL.  In this section, we discuss how  to conditionally restrict the expressive power of the generator or discriminator for QMMW and QuGAN can be conditionally restricted to tackle a given quantum information processing problem. 

For QMMW,  an extra `constraint' step should be involved in the update rule to restrict the expressive power. Naive QMMW is capable of approximating any quantum state without the imposition of any constraint,  as proved in Theorem \ref{thm2}.  The `constraint' step  ensures   that only the desired answer formulated in Section \ref{Sec:prob_app} can be efficiently approximated by QMMW. Two standard rules govern the design of the  `constraint' step, namely, that  it does not destroy the no-regret property of QMMW and that  it can be efficiently implemented by quantum operations. 

For QuGAN, the restriction of the expressive power can be achieved by adjusting the quantum circuit structure, so that only the desired answer formulated in Section \ref{Sec:prob_app} can be efficiently simulated. In particular, the arrangement of  the quantum gates of each block $U(\bm{\theta_i})$ and $U(\bm{\gamma_i})$  defined in Eqn.~(\ref{eqn:PQC_U_G}) should be redesigned. Although QuGAN cannot guarantee an  effective convergence rate as QMMW does, it may still have  quantum advantages, since QuGAN does not demand expensive measurements and can be efficiently implemented on near-term quantum devices.

 To facilitate understanding, we show how  to use QMMW and QuGAN can be used to  accomplish the entanglement test for a bipartite  pure state. The formal definition of  the separable bipartite pure state as follows \cite{horodecki2009quantum}.  Suppose that a given bipartite pure state $\ket{\Psi}_{AB}$ is represented by $N_A+N_B$ qubits, we say that  $\ket{\Psi}_{AB} \in\mathcal{H}_A\otimes\mathcal{H}_B$ is separable  if it can be written as	$\ket{\Psi}_{AB}= \ket{\phi}_{A}\ket{\phi}_{B}$ with $\ket{\phi}_{A}\in\mathcal{H}_A$ ($\ket{\phi}_{B}\in\mathcal{H}_B$).  

When QMMW is employed to distinguish entanglement from a bipartite pure state, we impose an  `constraint' step in updating the generated state. We define the target state as $\rho_{AB}=\ket{\Psi}\bra{\Psi}_{AB}$ at each training rounds, and two copies of $\sigma_G$ are generated as defined in Eqn.~(\ref{eqn:upd_G}). The `constrained' step refers to a partial trace step, i.e., by partial trace system $A$ for the first copy and system $B$ for the second copy, we have the product state $\Tr_A{(\sigma_G)}\otimes\Tr_B{(\sigma_G)}$. The integration of the `constraint' step and naive QMMW naturally results in Nash equilibrium  being  reached if and only if the input state is separable, since the generated state must be separable. Meanwhile, the `constraint' step satisfies the two standards rules. It is easy to prove that the no-regret property of the varied QMMW is conserved. The partial trace can be executed with $\mathcal{O}(1)$ complexity.    
  
When QuGAN is employed to distinguish entanglement from a bipartite pure state, we redesign the arrangement of quantum gates in each block of $U_G$ defined in Eqn.~(\ref{eqn:PQC_U_G}). No CNOT gate exists whose controlled qubit is in system $A$ and whose target qubit is in system $B$. The detailed quantum circuit architecture is shown in the right panel of Figure \ref{fig:Q-DC-GAN}. The modified quantum circuit structure indicates that Nash equilibrium can be reached if and only if the input state is separable,  since $U_G$ can only generate  a separable state.

 \begin{figure}[H]
\centering
	\begin{adjustbox}{width=0.96\textwidth}
	\begin{tikzcd}[row sep={0.8cm,between origins}]
	\lstick[wires=3]{$\ket{0}^{\otimes N_A}$}  & \qw    &	\gate[style={fill=pink!20},wires=6]{U_G/U_{\rho}}	&\gategroup[wires=7,steps=1,style={dashed,rounded corners,fill=blue!20, inner xsep=2pt},background]{}\gate[style={fill=yellow!20},wires=7]{U_D}  &\qw \\
											   &\vdots      &  \hphantom{wide label}	&\vdots	 &\qw \\
											   & \qw    &\qw 	&\qw					 &\qw	\\
	\lstick[wires=3]{$\ket{0}^{\otimes N_B}$}  & \qw   &\qw 	&\qw \\
											   & \vdots   	  &		&\hphantom{wide label}	& \\
											   & \qw    &\qw	&\qw &\qw \\
	\lstick{$\ket{0}$} 						   & \qw     &\qw  &\qw  &\qw &\meter{0}
		\end{tikzcd}

	\begin{tikzcd}[row sep={0.8cm,between origins}]
	\lstick[wires=3]{$\ket{0}^{\otimes N_A}$}  & \qw    &\gategroup[wires=6,steps=5,style={dashed,rounded corners,fill=pink!20, inner xsep=2pt},background]{$U_G(\bm{\theta})$} \gate[style={fill=yellow!20}]{U} & \ctrl{2}  &\gate[style={fill=yellow!20}]{U}   &\ctrl{2}	&\gategroup[wires=7,steps=5,style={dashed,rounded corners,fill=blue!20, inner xsep=2pt},background]{$U_D(\bm{\theta})$}\gate[style={fill=yellow!20}]{U} &\ctrl{2} &\qw &\qw &\qw &\qw \\
											   &\vdots    &  & &  &  \hphantom{wide label}\hdots 	&\vdots	 &  & &  & \\
											   & \qw   &\gate[style={fill=yellow!20}]{U}  	&\targ{}  &\gate[style={fill=yellow!20}]{U}  &\targ{} &\gate[style={fill=yellow!20}]{U}				 &\targ{}  &\ctrl{1} &\qw &\qw	&\qw	\\
	\lstick[wires=3]{$\ket{0}^{\otimes N_B}$}  & \qw   &\gate[style={fill=yellow!20}]{U} 	&\ctrl{2}  &\gate[style={fill=yellow!20}]{U}  &\qw &\gate[style={fill=yellow!20}]{U}  &\qw &\targ{} &\qw &\qw &\qw \\
											 &\vdots    &  & &  &  \hphantom{wide label}\hdots 	&\vdots	 &  &  &\hdots & &\\
											   & \qw    &\gate[style={fill=yellow!20}]{U} 	&\targ{}  &\gate[style={fill=yellow!20}]{U}  &\qw &\gate[style={fill=yellow!20}]{U} &\qw &\qw &\qw &\ctrl{1} &\qw\\
	\lstick{$\ket{0}$} 						   & \qw     &\qw  &\qw  &\qw &\qw &\gate[style={fill=yellow!20}]{U} &\qw &\qw &\qw &\targ{} &\meter{0}
		\end{tikzcd}

	\end{adjustbox}
		\caption{\small{The quantum circuit of QuGAN to accomplish the entanglement test for bipartite pure states. In the left panel,  $U_{\rho}$ (or $U_G$) is selected to to produce the real (or generated) state with prior $P(R)$ (or $P(G)$. In the right panel,  the circuit architecture of $U_G$ and $U_D$ is expanded. The notation $U$ is defined as $U=R_X\circ R_Y \circ R_Z$, where $R_X$, $R_Y$, and $R_Z$ are trainable parameterized single qubit gates along $X$, $Y$, $Z$ axis.}}
\label{fig:Q-DC-GAN}
\end{figure}
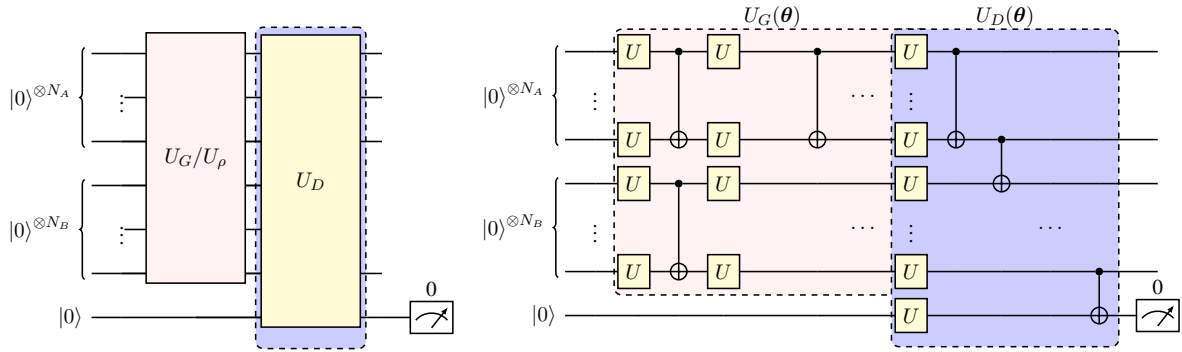
 
 It is valuable to compare QuGAN with another advanced method, the self-testing  method \cite{harrow2013testing}, which can also accomplish the pure  entanglement test task. The core ingredient of self-testing is the controlled-SWAP test; however, the controlled-swap test has several disadvantages under near-term devices \cite{benedetti2018adversarial}. To perform the self-testing method,  $2N+1$ qubits are required,  two copies should be accessed  simultaneously, and the ability to conduct nontrivial controlled gates and error correction is required.  In contrast to the self-testing method, QuGAN can flexibly  select the number of controlled gates, which is more suitable for near-term quantum devices.  

\section{numerical simulations}\label{Sec:Simu}
QMMW is a powerful tool for approximating  a given state. We  validate its performance by approximating 
a separable mixed state $\rho_{sep}=\frac{1}{2}\ket{0000}\bra{0000}+\frac{1}{2}\ket{1111}\bra{1111}$. The total number of training rounds is set as $T=400$ and $T=1600$, respectively. As  illustrated in Figure \ref{fig:QMMW}, the final training loss for $T = 400$ is $0.561$ with fidelity of  $0.929$. The final training loss for $T = 1600$ is $0.532$ with fidelity of $0.965$. The simulation results  indicate that the training loss rapidly  converges to the equilibrium value and the fidelity between the generated state and $\rho_{sep}$ tends to be $1$ with increased $T$. The simulation results are in accordance with the conclusion of Theorem \ref{thm2}, where the theoretical results are  $1/2+3\sqrt{4/400}=0.7$ and  $1/2+3\sqrt{4/1600}=0.65$, respectively.  The numerical simulations are implemented in  Python in conjunction with  QuTiP \cite{johansson2012qutip}.
 
\begin{figure}[h]
\centering
\includegraphics[width=0.9\textwidth]{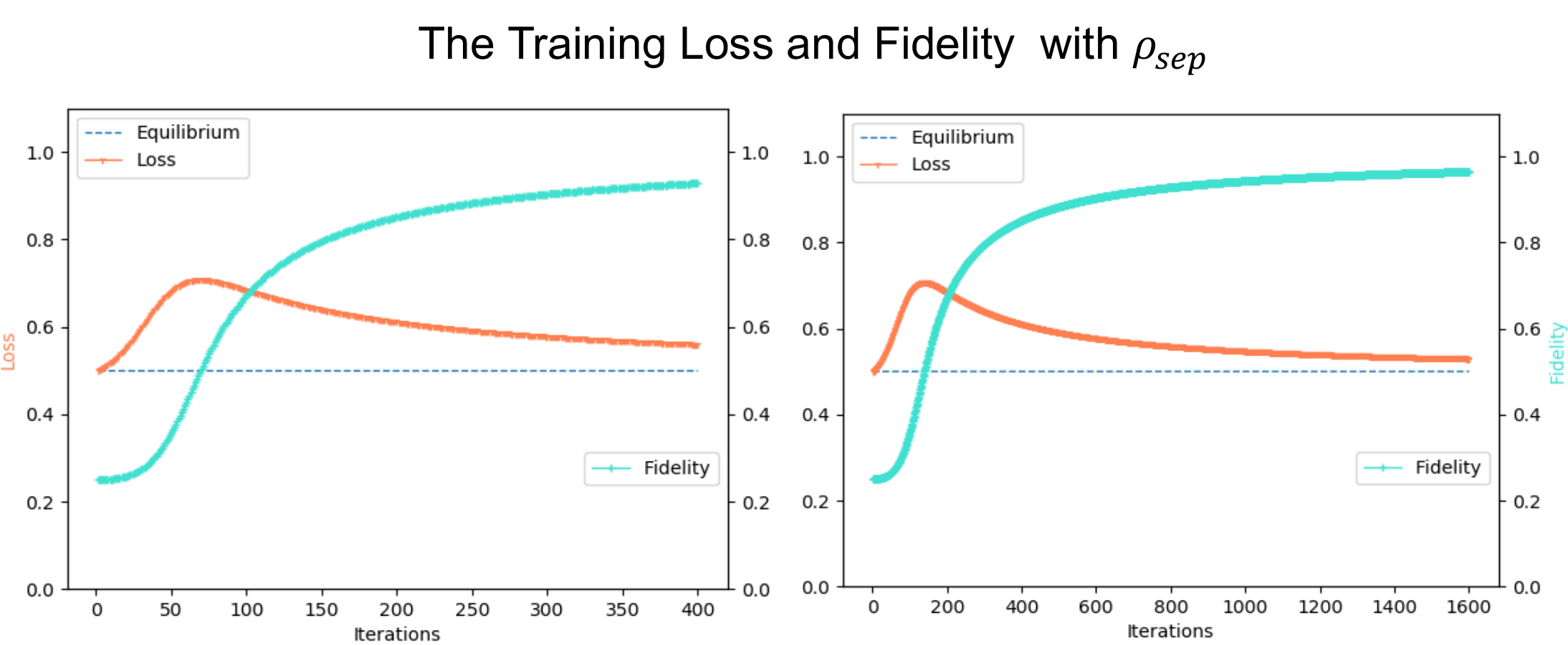}
\caption[small]{The left panel is the simulation result of QMMW with setting $T=400$. The right panel is the simulation result of QMMW with setting $T=1600$. }
\label{fig:QMMW}
\end{figure}

 We then benchmark the performance of the QuGANs to accomplish the entanglement test for bipartite pure states. The detailed procedure for constructing QuGAN is as follows. The trainable parameters  $\bm{\theta}$ (for $U_G$) and $\bm{\gamma}$ (for $U_D$) are randomly initialized and updated by the  zero-order differential method  \cite{mitarai2018quantum}. We set the total number of  training rounds $T$ as $500$. The prior defined in Eqn.~(\ref{eqn:5}) is set as $P(G)=P(R)=1/2$. The detailed quantum circuit structure is illustrated in Figure {\ref{fig:Q-DC-GAN}}. The number of blocks required to implement $U_G$ and $U_D$ as defined in Eqn.~(\ref{eqn:PQC_U_G}) is set as $L_1=7$ and $L_2=3$, respectively. The expressive power of $U_G$ is constrained as explained in Section \ref{Sec:QMM_Qu_app}. The quantum circuit architecture is demonstrated in Figure \ref{fig:Q-DC-GAN}. All numerical simulations are implemented in Python in conjunction with the PyQuil library \cite{smith2016practical}.  
 
 We now employ QuGAN to accomplish the entanglement test for two bipartite pure states, i.e., a separable state  $\ket{\Psi}=(\ket{00}_A+\ket{10}_A)\otimes\ket{00}_B/\sqrt{2}$ and an entangled state $\ket{\text{GHZ}}=(\ket{00}_A\otimes \ket{00}_B+\ket{11}_A\otimes \ket{11}_B)/\sqrt{2}$, where $A$ and $B$ refer to the bipartite system.  When the input state is separable state $\ket{\Psi}$,  the training loss oscillates around the optimal value after around $100$ steps and ranges  from $0.444$ to $0.559$, as shown in the outer plot. The corresponding fidelity between the target state and the generated state is always larger than $0.702$. The  training loss for the entanglement state case is far away from the optimal value, which oscillates  around $0.850$ after $300$ steps, as shown in the inner plot. The fidelity between the generated state and the given state $\ket{\text{GHZ}}$ is always below $0.250$. The simulation results echo  the analysis in Section \ref{Sec:QMM_Qu_app}. The simulation results are illustrated in Figure \ref{fig:GHZ_loss}. To accomplish the simulation,  QuGAN requires $143$ single and two qubit quantum gates, while self-testing method requires $240$ quantum gates.

\begin{figure}[H]
\centering
\includegraphics[width=0.5\textwidth]{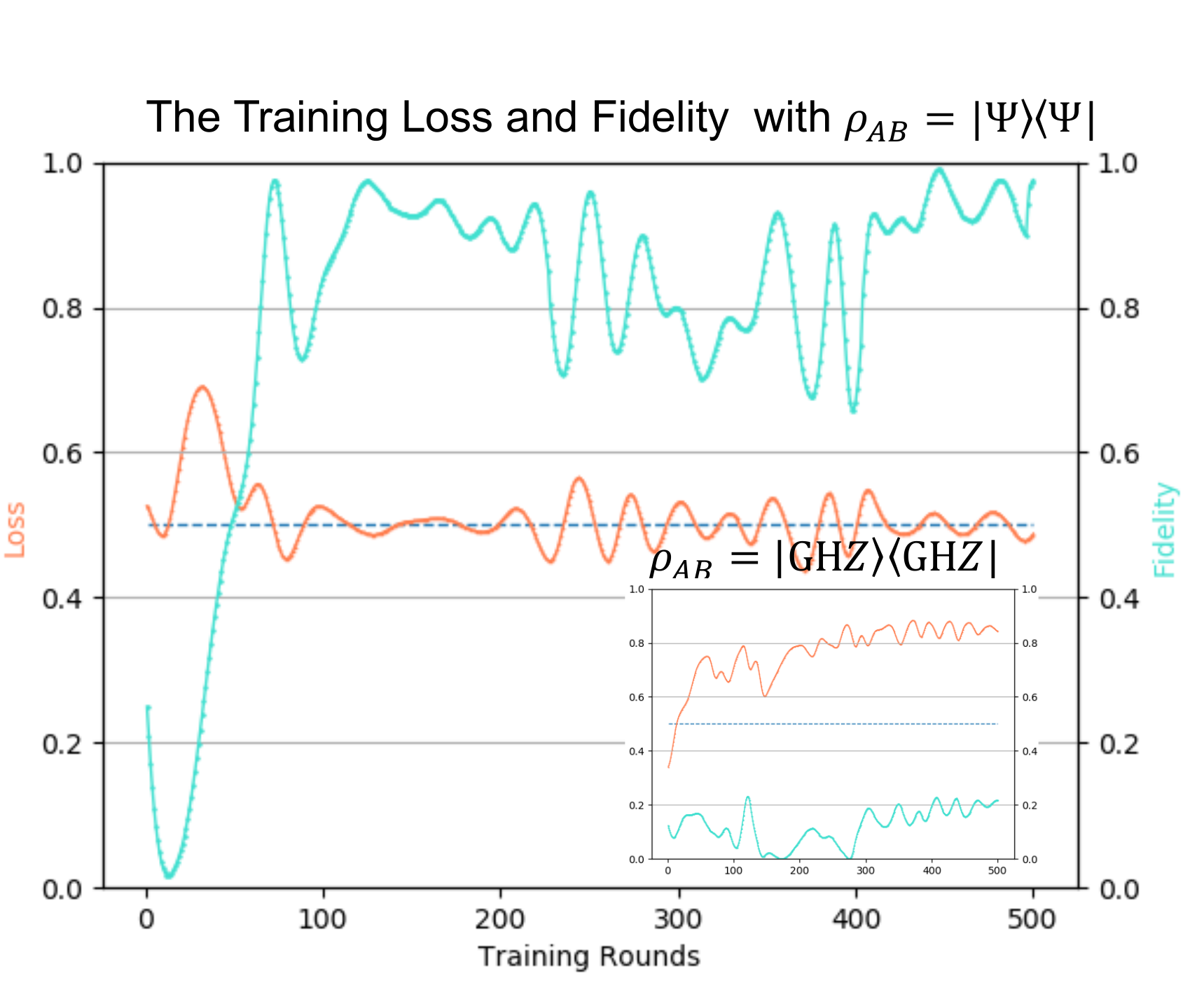}
\caption[small]{The outer plot is the simulation result of  QuGAN when the input is $\ket{\Psi}$. The inner plot is the simulation result of QuGAN when the input is $\ket{\text{GHZ}}$. }
\label{fig:GHZ_loss}
\end{figure}

\section{conclusion}\label{Sec:concl}
In this paper, we have presented the first attempt to approach quantum information processing problems by employing QuGAL. We have proposed  a general framework that enables   quantum information processing problems to be  tackled by using QuGAL. A major  advantage of QuGAL is its capability to  process quantum data directly, where  the required number of quantum  measurements is irrelevant to the size of quantum state. This advantage is significant in accomplishing  quantum information process tasks, since conventional methods generally demand exponential measurements to extract sufficient quantum information into a classical form.  

Encouraged  by the similarity between  QuGAL, online learning, and zero-sum game, we have exploited advanced  online learning methods to conquer two  issues in QuGAL, i.e., finding a quantum generative adversarial  learning algorithm that can rapidly converge to Nash equilibrium, and how the performance in training QuGANs QuGAN can be improved. To resolve the former issue, we proposed QMMW and proved that its  training loss can effectively  converge to Nash equilibrium with the increased number of training rounds. The computational complexity of QMMW is polynomially proportional to the number of qubits and training rounds. To solve the latter issue, we introduced the multiplicative weight training method.  The proposed method has the ability to relieve the dilemma encountered in training QuGANs such that the   optimization may get stuck in local minima. 

Lastly, we have described how to apply QMMW and QuGANs to solve quantum information processing tasks. We have shown that QMMW and QuGANs can be employed to accomplish entanglement test task for pure states. Several numerical simulations were conducted to  validate that QuGANs is capable of accomplishing  the entanglement test with modest quantum resources.   

Our future work has  two key directions. First, we will focus on applying QMMW and QuGANs to tackle more fundamental quantum information problems, e.g.,  the identification of  quantum correlation. Second, we will investigate  whether other advanced online learning methods exist that may improve the training performance of QuGAN. We believe that combining QuGAL with quantum information processing will benefit the fields of  quantum machine learning and quantum information.

\newpage
\bibliographystyle{plain}
\bibliography{myref2}

\begin{thebibliography}{10}

\bibitem{aaronson2018online}
Scott Aaronson, Xinyi Chen, Elad Hazan, Satyen Kale, and Ashwin Nayak.
\newblock Online learning of quantum states.
\newblock In {\em Advances in Neural Information Processing Systems}, pages
  8962--8972, 2018.

\bibitem{arjovsky2017wasserstein}
Martin Arjovsky, Soumith Chintala, and L{\'e}on Bottou.
\newblock Wasserstein generative adversarial networks.
\newblock In {\em International Conference on Machine Learning}, pages
  214--223, 2017.

\bibitem{barnett2009quantum}
Stephen~M Barnett and Sarah Croke.
\newblock Quantum state discrimination.
\newblock {\em Advances in Optics and Photonics}, 1(2):238--278, 2009.

\bibitem{benedetti2018adversarial}
Marcello Benedetti, Edward Grant, Leonard Wossnig, and Simone Severini.
\newblock Adversarial quantum circuit learning for pure state approximation.
\newblock {\em arXiv preprint arXiv:1806.00463}, 2018.

\bibitem{biamonte2017quantum}
Jacob Biamonte, Peter Wittek, Nicola Pancotti, Patrick Rebentrost, Nathan
  Wiebe, and Seth Lloyd.
\newblock Quantum machine learning.
\newblock {\em Nature}, 549(7671):195, 2017.

\bibitem{boyd2004convex}
Stephen Boyd and Lieven Vandenberghe.
\newblock {\em Convex optimization}.
\newblock Cambridge University Press, 2004.

\bibitem{2017arXiv171002581B}
Fernando G.~S.~L. {Brand{\~a}o}, Amir {Kalev}, Tongyang {Li}, Cedric {Yen-Yu
  Lin}, Krysta~M. {Svore}, and Xiaodi {Wu}.
\newblock {Quantum SDP solvers: Large speed-ups, optimality, and applications
  to Quantum Learning}.
\newblock {\em arXiv e-prints}, page arXiv:1710.02581, October 2017.

\bibitem{buhrman2001quantum}
Harry Buhrman, Richard Cleve, John Watrous, and Ronald De~Wolf.
\newblock Quantum fingerprinting.
\newblock {\em Physical Review Letters}, 87(16):167902, 2001.

\bibitem{carrasquilla2017machine}
Juan Carrasquilla and Roger~G Melko.
\newblock Machine learning phases of matter.
\newblock {\em Nature Physics}, 13(5):431, 2017.

\bibitem{chitambar2014local}
Eric Chitambar, Runyao Duan, and Min-Hsiu Hsieh.
\newblock When do local operations and classical communication suffice for
  two-qubit state discrimination?
\newblock {\em IEEE Transactions on Information Theory}, 60(3):1549--1561,
  2014.

\bibitem{chitambar2013revisiting}
Eric Chitambar and Min-Hsiu Hsieh.
\newblock Revisiting the optimal detection of quantum information.
\newblock {\em Physical Review A}, 88(2):020302, 2013.

\bibitem{chitambar2014asymptotic}
Eric Chitambar and Min-Hsiu Hsieh.
\newblock Asymptotic state discrimination and a strict hierarchy in
  distinguishability norms.
\newblock {\em Journal of Mathematical Physics}, 55(11):112204, 2014.

\bibitem{chitambar2017round}
Eric Chitambar and Min-Hsiu Hsieh.
\newblock Round complexity in the local transformations of quantum and
  classical states.
\newblock {\em Nature {C}ommunications}, 8(1):2086, 2017.

\bibitem{dallaire2018quantum}
Pierre-Luc Dallaire-Demers and Nathan Killoran.
\newblock Quantum generative adversarial networks.
\newblock {\em arXiv preprint arXiv:1804.08641}, 2018.

\bibitem{doherty2004complete}
Andrew~C Doherty, Pablo~A Parrilo, and Federico~M Spedalieri.
\newblock Complete family of separability criteria.
\newblock {\em Physical Review A}, 69(2):022308, 2004.

\bibitem{du2018expressive}
Yuxuan Du, Min-Hsiu Hsieh, Tongliang Liu, and Dacheng Tao.
\newblock The expressive power of parameterized quantum circuits.
\newblock {\em arXiv preprint arXiv:1810.11922}, 2018.

\bibitem{farina2017regret}
Gabriele Farina, Christian Kroer, and Tuomas Sandholm.
\newblock Regret minimization in behaviorally-constrained zero-sum games.
\newblock In {\em Proceedings of the 34th International Conference on Machine
  Learning-Volume 70}, pages 1107--1116. JMLR. org, 2017.

\bibitem{goodfellow2016nips}
Ian Goodfellow.
\newblock {NIPS} 2016 tutorial: Generative adversarial networks.
\newblock {\em arXiv preprint arXiv:1701.00160}, 2016.

\bibitem{goodfellow2014generative}
Ian Goodfellow, Jean Pouget-Abadie, Mehdi Mirza, Bing Xu, David Warde-Farley,
  Sherjil Ozair, Aaron Courville, and Yoshua Bengio.
\newblock Generative adversarial nets.
\newblock In {\em Advances in {N}eural {I}nformation {P}rocessing {S}ystems},
  pages 2672--2680, 2014.

\bibitem{grnarova2018an}
Paulina Grnarova, Kfir~Y Levy, Aurelien Lucchi, Thomas Hofmann, and Andreas
  Krause.
\newblock An online learning approach to generative adversarial networks.
\newblock In {\em International Conference on Learning Representations}, 2018.

\bibitem{gutoski2013parallel}
Gus Gutoski and Xiaodi Wu.
\newblock Parallel approximation of min-max problems.
\newblock {\em {C}omputational {C}omplexity}, 22(2):385--428, 2013.

\bibitem{harrow2013testing}
Aram~W Harrow and Ashley Montanaro.
\newblock Testing product states, quantum {M}erlin-{A}rthur games and tensor
  optimization.
\newblock {\em Journal of the ACM (JACM)}, 60(1):3, 2013.

\bibitem{hazan2016introduction}
Elad Hazan et~al.
\newblock Introduction to online convex optimization.
\newblock {\em Foundations and Trends{\textregistered} in Optimization},
  2(3-4):157--325, 2016.

\bibitem{horodecki1997separability}
Pawel Horodecki.
\newblock Separability criterion and inseparable mixed states with positive
  partial transposition.
\newblock {\em arXiv preprint quant-ph/9703004}, 1997.

\bibitem{horodecki2009quantum}
Ryszard Horodecki, Pawe{\l} Horodecki, Micha{\l} Horodecki, and Karol
  Horodecki.
\newblock Quantum entanglement.
\newblock {\em Reviews of {M}odern {P}hysics}, 81(2):865, 2009.

\bibitem{jain2009parallel}
Rahul Jain and John Watrous.
\newblock Parallel approximation of non-interactive zero-sum quantum games.
\newblock In {\em 2009 24th Annual IEEE Conference on Computational
  Complexity}, pages 243--253. IEEE, 2009.

\bibitem{johansson2012qutip}
J~Robert Johansson, Paul~D Nation, and Franco Nori.
\newblock Qutip: An open-source python framework for the dynamics of open
  quantum systems.
\newblock {\em Computer Physics Communications}, 183(8):1760--1772, 2012.

\bibitem{kale2007efficient}
S.~Kale.
\newblock {\em Efficient algorithms using the multiplicative weights update
  method}.
\newblock Princeton University, 2007.

\bibitem{lloyd2018quantum}
Seth Lloyd and Christian Weedbrook.
\newblock Quantum generative adversarial learning.
\newblock {\em arXiv preprint arXiv:1804.09139}, 2018.

\bibitem{lu2018separability}
Sirui Lu, Shilin Huang, Keren Li, Jun Li, Jianxin Chen, Dawei Lu, Zhengfeng Ji,
  Yi~Shen, Duanlu Zhou, and Bei Zeng.
\newblock Separability-entanglement classifier via machine learning.
\newblock {\em Physical Review A}, 98(1), 2018.

\bibitem{ma2018transforming}
Yue-Chi Ma and Man-Hong Yung.
\newblock Transforming {B}ell’s inequalities into state classifiers with
  machine learning.
\newblock {\em npj Quantum Information}, 4(1):34, 2018.

\bibitem{mitarai2018quantum}
Kosuke Mitarai, Makoto Negoro, Masahiro Kitagawa, and Keisuke Fujii.
\newblock Quantum circuit learning.
\newblock {\em arXiv preprint arXiv:1803.00745}, 2018.

\bibitem{nielsen2010quantum}
Michael~A Nielsen and Isaac~L Chuang.
\newblock {\em Quantum computation and quantum information}.
\newblock Cambridge University Press, 2010.

\bibitem{osborne1994course}
Martin~J Osborne and Ariel Rubinstein.
\newblock {\em A course in game theory}.
\newblock {MIT} {P}ress, 1994.

\bibitem{pantazis2018training}
Yannis Pantazis, Dipjyoti Paul, Michail Fasoulakis, and Yannis Stylianou.
\newblock Training generative adversarial networks with weights.
\newblock {\em arXiv preprint arXiv:1811.02598}, 2018.

\bibitem{preskill2018quantum}
John Preskill.
\newblock Quantum computing in the {NISQ} era and beyond.
\newblock {\em arXiv preprint arXiv:1801.00862}, 2018.

\bibitem{romero2019variational}
Jonathan Romero and Alan Aspuru-Guzik.
\newblock Variational quantum generators: Generative adversarial quantum
  machine learning for continuous distributions.
\newblock {\em arXiv preprint arXiv:1901.00848}, 2019.

\bibitem{shor1999polynomial}
Peter~W Shor.
\newblock Polynomial-time algorithms for prime factorization and discrete
  logarithms on a quantum computer.
\newblock {\em {SIAM} {R}eview}, 41(2):303--332, 1999.

\bibitem{situ2018adversarial}
Haozhen Situ, Zhimin He, Lvzhou Li, and Shenggen Zheng.
\newblock Adversarial training of quantum born machine.
\newblock {\em arXiv preprint arXiv:1807.01235}, 2018.

\bibitem{smith2016practical}
Robert~S Smith, Michael~J Curtis, and William~J Zeng.
\newblock A practical quantum instruction set architecture.
\newblock {\em arXiv preprint arXiv:1608.03355}, 2016.

\bibitem{van2018improvements}
Joran van Apeldoorn and Andr{\'a}s Gily{\'e}n.
\newblock Improvements in quantum {SDP}-solving with applications.
\newblock {\em arXiv preprint arXiv:1804.05058}, 2018.

\bibitem{van2019quantum}
Joran van Apeldoorn and Andr{\'a}s Gily{\'e}n.
\newblock Quantum algorithms for zero-sum games.
\newblock {\em arXiv preprint arXiv:1904.03180}, 2019.

\bibitem{van2017learning}
Evert~PL Van~Nieuwenburg, Ye-Hua Liu, and Sebastian~D Huber.
\newblock Learning phase transitions by confusion.
\newblock {\em Nature Physics}, 13(5):435, 2017.

\bibitem{witten2016data}
Ian~H Witten, Eibe Frank, Mark~A Hall, and Christopher~J Pal.
\newblock {\em Data mining: Practical machine learning tools and techniques}.
\newblock Morgan Kaufmann, 2016.

\bibitem{zeng2018learning}
Jinfeng Zeng, Yufeng Wu, Jin-Guo Liu, Lei Wang, and Jiangping Hu.
\newblock Learning and inference on generative adversarial quantum circuits.
\newblock {\em arXiv preprint arXiv:1808.03425}, 2018.

\bibitem{zhang2016understanding}
Chiyuan Zhang, Samy Bengio, Moritz Hardt, Benjamin Recht, and Oriol Vinyals.
\newblock Understanding deep learning requires rethinking generalization.
\newblock {\em arXiv preprint arXiv:1611.03530}, 2016.

\bibitem{zoufal2019quantum}
Christa Zoufal, Aur{\'e}lien Lucchi, and Stefan Woerner.
\newblock Quantum generative adversarial networks for learning and loading
  random distributions.
\newblock {\em arXiv preprint arXiv:1904.00043}, 2019.

\end{thebibliography}

\newpage

\appendix
\section{Supplemental Materials} 

 \subsection{SM(A) Proof of Theorem 1}
  
  The analysis of the convergence of quantum multiplicative matrix weight (QMMW) relies mainly on conclusions drawn from game theory and online learning. To provide an illustrative proof,  we first introduce the necessary concepts from these two fields and then build the connection to Theorem 1.  
   
The formal definition of the zero-sum game is:  \begin{definition}[Zero-sum Game, \cite{farina2017regret}]
  A two-player zero-sum game is a tuple $(\mathcal{X}, \mathcal{Y}, u)$ where $\mathcal{X}$ represents the finite set of actions that player $1$ can play, $\mathcal{Y}$ represents the finite set of actions that player $2$ can play, and $u : \mathcal{X} \times \mathcal{Y} \rightarrow R$ is the payoff function for player $1$, mapping the pair of actions $(x , y)\in  (\mathcal{X},\mathcal{Y})$ of the players into the payoff for player $1$. The corresponding payoff for player 2 is given by $-u(x, y)$.
  \end{definition}
  Exploiting the definition of the zero-sum game, we  introduce two concepts in game theory. The first concept is \textit{approximated best response}, defined as follows:
  \begin{definition}[Approximate best response]\label{def:ABR}
 Given a zero-sum game $(\mathcal{X}, \mathcal{Y}, u)$, we say that $x\in\mathcal{X}$ is an $\varepsilon$-best response to $y$ for player $1$ if $u(x, y) +\varepsilon \geq u(\hat{x}, y)$ for all $ \hat{x}\in\mathcal{X}$. Symmetrically, given $x\in\mathcal{X}$, we say that $y \in \mathcal{Y}$ is an $\varepsilon$-best response to $x$ for player $2$ if $-u(x, y) + \varepsilon \geq -u(x, \hat{y})$ for all $\hat{y} \in Y$. 
  \end{definition}
  Another concept is \textit{approximated Nash equilibrium}, defined as:  \begin{definition}[Approximate Nash equilibrium]\label{def:APN}
 Given a zero-sum game $(\mathcal{X}, \mathcal{Y}, u)$, the strategy pair $(x,y)\in (\mathcal{X}\times \mathcal{Y})$ is an $\varepsilon$-Nash equilibrium for the game if $x$ is an $\varepsilon$-best response to $y$ for player $1$, and $y$ is an $\varepsilon$-best response to $x$ for player 2. Note that Nash equilibrium can be treated as a $0$-Nash equilibrium.
  \end{definition}
A well-known conclusion  between regret and approximate Nash equilibria is as follows  \cite{farina2017regret}:
\begin{prop}\label{prop:2}
	In a zero-sum game, if the average  regrets of the players up to step $T$ are such that $R_T (\hat{x}):=\sum_{t=1}^T(u(x_t,y_t)-u(\hat{x},y_t))$ with $R_T (\hat{x})/T \leq \varepsilon_1$, $R _T (\hat{y}):=\sum_{t=1}^T(-u(x_t,y_t)+u(x_t,\hat{y}_t))$ with $R _T (\hat{y})/T \leq \varepsilon_2$, for all actions $\hat{x}\in\mathcal{X}$, $ \hat{y} \in \mathcal{Y}$, then the strategy pair $(\bar{x}_T,\bar{y}_T):=(\sum_{t=1}^Tx_t/T, \sum_{t=1}^Ty_t/T)\in (\mathcal{X},\mathcal{Y})$ is a $(\varepsilon_1+\varepsilon_2)$-Nash equilibrium.
\end{prop}
We now connect the Proposition \ref{prop:2} with Theorem 1. In QMMW, the minimized regret for the generator is   
\begin{equation}
	R_T(\sigma_G) = {-\sum_{t=1}^T\mathcal{L}(\sigma_G^{(t)},\sigma_D^{(t)})+\min_{\sigma_G} \sum_{t=1}^T\mathcal{L}(\sigma_G,\sigma_D^{(t)})}~.
\end{equation}
The minimized regret for the discriminator is
\begin{equation}
R_T(\sigma_D)= {\sum_{t=1}^T\mathcal{L}(\sigma_G^{(t)},\sigma_D^{(t)})-\min_{\sigma_D} \sum_{t=1}^T\mathcal{L}(\sigma_G^{(t)},\sigma_D)}~.	
\end{equation}
Suppose that  the optimal strategy pair is $(\sigma_G^*, \sigma_D^{*})$ and the corresponding Nash equilibrium is $\mathcal{L}(\sigma_G^*, \sigma_D^{*})=1/2$. Following the statement of Proposition \ref{prop:2} and Definition \ref{def:ABR}, with setting  $\varepsilon_G$ and $\varepsilon_G$ that satisfies   $R_T(\sigma_G)/T\leq \varepsilon_G$ and $R_T(\sigma_D)/T\leq \varepsilon_D$,  the  strategy pair $(\bar{\sigma}_G, \bar{\sigma}_D)$ with $\bar{\sigma}=\sum_{t=1}^T\sigma^{(t)}/{T}$  and the   $\bar{\sigma}_D=\sum_{t=1}^T{\sigma_D}^{(t)}/T$, is a $(\varepsilon_G+\varepsilon_D)$-Nash equilibrium, i.e.,
\begin{equation}\label{eqn:thm3_1}
	\mathcal{L}({\sigma^*}, \bar{\sigma}_D)-(\varepsilon_G+\varepsilon_D)\leq \mathcal{L}(\bar{\sigma}_G, \bar{\sigma}_D)
\leq \mathcal{L}(\bar{\sigma}_G, {\sigma_D^*}) +(\varepsilon_G+\varepsilon_D)~.
\end{equation}
These two inequalities come from Definition \ref{def:ABR}. 
According to the definition of Nash equilibrium, which is an $0$-Nash Equilibrium with an optimal  strategy pair $(\sigma^*, \sigma_D^{*})$, we rewrite  Eqn.~(\ref{eqn:thm3_1}) as  
\begin{equation}\label{eqn:thm3_2}
	\mathcal{L}(\bar{\sigma}_G, \sigma_D^{*})
\leq \mathcal{L}(\sigma_G^*, \sigma_D^{*})
\leq \mathcal{L}(\sigma_G^*, \bar{\sigma_D}) ~.
\end{equation}
Connecting  Eqn.~(\ref{eqn:thm3_1}) with  Eqn.~(\ref{eqn:thm3_2}), we have 
\begin{equation}\label{eqn:10-1}
		| \mathcal{L}(\bar{\sigma}_G, \bar{\sigma}_D)
-  \mathcal{L}(\sigma^*_G, \sigma_D^{*}))
| \leq  \varepsilon_G+\varepsilon_D~.
\end{equation} 
The goal of QMMW, or Theorem 1, is to prove  $\varepsilon_G+\varepsilon_D\leq 3\sqrt{N/T}$, which implies that $\mathcal{L}(\bar{\sigma}_G, \bar{\sigma}_D)$ converges to the optimal value $\mathcal{L}(\sigma_G^*, \sigma_D^*)=1/2$ with increasing $T$. 

 \begin{proof}[Proof of Theorem 1]
 
As discussed above, we aim to prove  that QMMW possesses the no-regret property when the original state $\rho$ is separable. 

We employ the following two claims to quantify two regrets $R_T(\sigma_G)$ and $R_T(\sigma_D)$ (Proof of Claim \ref{clm6} and  Claim \ref{clm1} provided later in this section.):
\begin{clam}\label{clm6}
The regret of $R_T(\sigma_G)$  based on the update rule of the Quantum Matrix Multiplicative Weights algorithm (to maximize the loss) is bounded by $\varepsilon_G$, i.e.,
\begin{equation}
\frac{R_T(\sigma_G)}{T}\leq \varepsilon_G=\frac{\epsilon^2 T + N}{2\epsilon T}~.
\end{equation} 
\end{clam}
\begin{clam}\label{clm1}
The regret of $R_T(\sigma_D)$  based on the update rule of the Quantum Matrix Multiplicative Weights algorithm (to minimize the loss) is bounded by $\varepsilon_D$, i.e.,
\begin{equation}
	\frac{R_T(\sigma_D)}{T}\leq \varepsilon_D=\frac{\epsilon^2 T + (N+1)}{2\epsilon T}~.
\end{equation}
\end{clam}
Combining the above two claims and with setting $\epsilon=2\sqrt{N/T}$, we have
\begin{equation}
	\varepsilon_G+\varepsilon_D\leq   \frac{2\epsilon^2 T+2(N+1)}{2\epsilon T}\leq 3\sqrt{\frac{N}{T}}~,
\end{equation}
where the first inequality comes from the results of Claim \ref{clm6} and Claim \ref{clm1}, and the second equality results from  $N+1\leq 2N$  and the definition of $\epsilon$. The discrepancy between the loss $\mathcal{L}(\bar{\sigma}_G, \bar{\sigma}_D)$ and the equilibrium value $\mathcal{L}({\sigma}_G^*, {\sigma}_D^*)$ yields 
\begin{equation}\label{eqn:10-2}
	\left|\mathcal{L}(\bar{\sigma}_G, \bar{\sigma}_{D})-\mathcal{L}(\sigma_G^*, \sigma_D^*)\right|\leq 3\sqrt{\frac{N}{T}}.
\end{equation}
\end{proof}
Before giving the proof of Claim \ref{clm1} and Claim \ref{clm6}, we introduce the following two results to facilitate the proof. 
\begin{coro}[Corollary 2, \cite{kale2007efficient}]\label{coro:clm7}
 For any $\epsilon\leq 1$, let $\epsilon_1 =1 -e^{-\epsilon}$ and $\epsilon_2=e^{\epsilon} - 1$. We then  have the following matrix inequalities:
 \begin{itemize}
 	\item If all eigenvalues of a symmetric matrix $A$ lie in $[0,1]$, then $e^{\epsilon A}\preceq \mathbb{I}-\epsilon_1A$;
 	\item If all eigenvalues of a symmetric matrix $A$ lie in $[-1,0]$, then $e^{\epsilon A}\preceq \mathbb{I}-\epsilon_2A$; 
 \end{itemize}	
\end{coro}

\begin{proof}[Proof of Claim \ref{clm6}]
Following the observation of Eqn.~(\ref{eqn:10-1}) and Proposition \ref{prop:2}, we  hope the regret $R_T(\sigma_G)$ possesses the no-regret property with $R_T(\sigma_G)\leq \varepsilon_G \sim O(1/\sqrt{T})$, i.e.,
\begin{equation}
 \sum_{t=1}^T\mathcal{L}(\sigma_G^{(t)}, \sigma_D^{(t)})  \geq \sum_{t=1}^T\mathcal{L}(\sigma_G^*, \sigma_D^{(t)}) - \varepsilon_GT~.
\end{equation}
To quantify $R_T(\sigma_{G})$, we define the potential function $\Theta^{(t)}$ and track its evolution with varying $t$, which is analogous to the proof of the conventional MMW algorithm, i.e., $\Theta^{(t)} =\Tr(W^{(t)}) $  and $W^{(t)}=e^{\sum_{\tau=1}^{t-1} -\epsilon \sigma_D^{(\tau)}}$.  By defining  $\epsilon_1=1-e^{-\epsilon}$, we have:
\begin{eqnarray}\label{eqn:clm6}
&&\Theta^{(t+1)} =\Tr{(W^{(t+1)})}=\Tr{(e^{\sum_{\tau=1}^t -\epsilon \sigma_D^{(\tau)}})} \nonumber\\
\leq &&\Tr{(e^{(\sum_{\tau=1}^{t-1} -\epsilon \sigma_D^{(\tau)})} e^{-\epsilon \sigma_D^{(t)}})} 
	=\Tr(W^{(t)}e^{-\epsilon \sigma_D^{(t)}})	\leq 
\Tr(W^{(t)}(\mathbb{I}-\epsilon_1\sigma_D^{(t)})) \nonumber\\
	=&&
  \Tr{(W^{(t)})}(1 -\epsilon_1\Tr(\sigma_D^{(t)}\frac{W^{(t)}}{\Tr{(W^{(t)})}}))       
\leq 
  \Theta^{(t)}e^{-\epsilon_1\Tr{(\sigma_D^{(t)}\frac{W^{(t)}}{\Tr{(W^{(t)})}})}}\nonumber\\
	= && 
 \Theta^{(t)}e^{-\epsilon_1\Tr{\left(\sigma_D^{(t)}\sigma_G^{(t)}\right)}}~.
\end{eqnarray}
 The first inequality comes from the  Golden-Thompson inequality, the second  inequality employs the conclusion of Corollary \ref{coro:clm7}, the third is supported by $1-x\leq e^{-x}$,  and the last equality arises in $\frac{W^{(t)}}{\Tr{(W^{(t)})}} = \sigma_G^{(t)}$. 

By induction, since  $\Theta^{(1)}= \Tr(e^{-\epsilon \sigma_D^{(1)}})\leq \Tr(e^{0}) \leq 2^{N}$, we have 
\begin{eqnarray}\label{eqn:211}
	\Theta^{(T+1)}\leq 2^{N}e^{-\epsilon_1\sum_{t=1}^T\Tr{\left(\sigma_D^{(t)}\sigma_G^{(t)}\right)}}.
\end{eqnarray}

The denotation  $A=\sum_{t=1}^{T}{(\sigma_D^{(t)})}$,    $\Theta^{(T+1)}$ yields $$\Theta^{(T+1)}=\Tr(\exp(-\epsilon A))=\sum_k \exp{(-\epsilon\lambda_k(A))}\geq \exp{(-\epsilon\lambda_n(A))}~,$$ where  $\lambda_n(A)$ refers to the minimum eigenvalue  of $A$. We also have $$e^{-\epsilon\sum_{t=1}^{T}\Tr{(\sigma_D^{(t)}\sigma_G^* )}}\leq \exp{(-\epsilon\lambda_n(A))}~,$$ since we always have $\lambda_n(A)\leq \Tr(A\sigma_G^*)$ with $\Tr(\sigma_G^*)=1$ and then $\exp({-\epsilon\lambda_n(A)}) \geq \exp({-\epsilon\Tr(A\sigma_G^*)})$. The lower bound of $\Theta^{(T+1)}$ is therefore 
\begin{equation}\label{eqn:21}
	\Theta^{(T+1)}\geq e^{-\epsilon\sum_{t=1}^T\Tr{(\sigma_D^{(t)}\sigma_G^*)}}~.
\end{equation}


By connecting Eqn.~(\ref{eqn:21}) with Eqn.~(\ref{eqn:211}), we have 
\begin{eqnarray}\label{eqn:22}
 &&e^{-\epsilon\sum_{t=1}^T\Tr{(\sigma_D^{(t)}\sigma_G^*)}}\leq  2^{N}  e^{-\epsilon_1\sum_{t=1}^T\Tr{(\sigma_D^{(t)}\sigma_G^{(t)})}}~.
\end{eqnarray}

Taking the logarithms of Eqn.~(\ref{eqn:22}) and exploiting $\epsilon_1\geq \epsilon(1-\epsilon)$, we obtain the  following inequality,
\begin{align*}\label{eqn:221}
&(1-\epsilon)\sum_{t=1}^T \Tr(\sigma_G^{(t)}\sigma_D^{(t)}) \leq \sum_{t=1}^T \Tr(\sigma_G^*\sigma_D^{(t)}) + \frac{N}{\epsilon}  \nonumber  \\
\Rightarrow &  \sum_{t=1}^T \Tr(\sigma_G^{(t)}\sigma_D^{(t)})  \leq \sum_{t=1}^T \Tr(\sigma_G^*\sigma_D^{(t)}) + \epsilon T + \frac{N}{\epsilon} \nonumber\\
\Rightarrow &  \sum_{t=1}^T \Tr((\rho\sigma_D^{(t)}) 
   -\sum_{t=1}^T \Tr(\sigma_G^{(t)}\sigma_D^{(t)}) \geq
\sum_{t=1}^T \Tr((\rho\sigma_D^{(t)})   -\sum_{t=1}^T \Tr(\sigma_G^*\sigma_D^{(t)})
  - \epsilon T - \frac{N}{\epsilon}\nonumber \\
  \Rightarrow & \sum_{t=1}^T\mathcal{L}(\sigma_G^{(t)}, \sigma_D^{(t)})
 \geq \sum_{t=1}^T\mathcal{L}(\sigma_G^*, \sigma_D^{(t)}) - \frac{\epsilon^2 T + N}{2\epsilon}~.
\end{align*}

The first arrow results from the fact that: $0\leq \epsilon \Tr(\sigma_G^{(t)}\sigma_D^{(t)})\leq \epsilon $ with  $\Tr(\sigma_G^{(t)}\sigma_D^{(t)})\leq \Tr(\sigma_G^{(t)})\Tr(\sigma_D^{(t)})=1$ for any $1\leq t\leq T$. The second arrow comes from adding  the term $\sum_{t=1}^T \Tr((\rho\sigma_D^{(t)})$ on both sides. The last arrow comes from the definition of the loss function of QMMW. The above equation indicates that $R_T(\sigma_G)\leq \varepsilon_G=\frac{\epsilon^2 T + N}{2\epsilon T}$.
\end{proof}

\begin{proof}[Proof of Claim \ref{clm1}] This claim can be easily proved by imitating the proof of Claim \ref{clm6}. Conceptually, we hope to derive a bound of the classification error, i.e.,
\begin{equation}
\sum_{t=1}^T\mathcal{L}(\sigma_G^{(t)}, \sigma_D^{(t)})  \leq \sum_{t=1}^T\mathcal{L}(\sigma_G^{(t)},  \sigma_D^{*}) +\varepsilon_D T  ~,
\end{equation}
where $\sigma_D^{*}$ refers to the optimal solution. We can then quantify the regret of $R_T( \mathcal{D}_{AB})$. 

Following the proof of Claim \ref{clm6}, we define  a potential function $\Theta^{(T)}=\Tr(W^{T})$ with $W^{(T)} =e^{-\epsilon \sum_{t=1}^{T-1} (\rho-\sigma_G^{(t)})}$.  
Let  $Q^{(t)}$ be $Q^{(t)}=\rho-\sigma_G^{(t)}$ with $Q^{(t)}=Q^{(t)}_++Q^{(t)}_-$. We decompose $Q(t)$ into two terms, $Q^{(t)}_+$ and  $Q^{(t)}_-$ with $Q(t)=Q^{(t)}_++Q^{(t)}_-$, where $Q^{(t)}_+$ ($Q^{(t)}_-$) is formed by all non-negative (negative) eigenvalues and eigenvectors of $Q^{(t)}$.   Defining  that   $\epsilon_1=1-e^{-\epsilon}$ and  $\epsilon_2=e^{\epsilon} - 1$, 
  we have:
\begin{eqnarray}\label{eqn:clm5}
&&\Theta^{(t+1)} =\Tr{(W^{(t+1)})} 
=\Tr{\left(e^{\sum_{\tau=1}^t -\epsilon (\rho-\sigma_G^{(\tau)} )}\right)} \nonumber\\
\leq &&\Tr{\left(e^{\sum_{\tau=1}^{t-1} -\epsilon (\rho-\sigma_G^{(\tau)})} e^{-\epsilon Q^{(t)} }\right)} = \Tr\left(W^{(t)}e^{-\epsilon (Q^{(t)}_{+} + Q^{(t)}_{-}) }\right) \nonumber \\
\leq &&  \Tr\left(W^{(t)}e^{-\epsilon Q^{(t)}_{+}  }e^{-\epsilon  Q^{(t)}_{-}}\right) \nonumber\\
		\leq && 
\Tr\left(W^{(t)}(\mathbb{I}-\epsilon_1Q^{(t)}_+)(\mathbb{I}-\epsilon_2Q^{(t)}_-)\right) \nonumber \\
= && 
\Tr\left(W^{(t)}(\mathbb{I}-\epsilon_1Q^{(t)}_+ - \epsilon_2Q^{(t)}_-)\right), \text{since $Q_+\perp Q_-$} \nonumber\\
= && 
\Tr\left(W^{(t)}\right)\left(1-\frac{\Tr\left(W^{(t)}(\epsilon_1Q^{(t)}_++\epsilon_2Q^{(t)}_-)\right)}{\Tr\left(W^{(t)}\right)}\right)  \nonumber \\
= &&
\Tr\left(W^{(t)}\right)\left(1-{\Tr\left(\sigma_D^{(t)}(\epsilon_1Q^{(t)}_++\epsilon_2Q^{(t)}_-)\right)}\right)  \nonumber\\
= && 
\Tr\left(W^{(t)}\right)\left(1-{\Tr\left(\sigma_D^{(t)}(\epsilon_1Q^{(t)}+(\epsilon_2-\epsilon_1)Q^{(t)}_-)\right)}\right)  \nonumber\\
\leq  && 
\Tr\left(W^{(t)}\right)\left(1-{\Tr\left(\sigma_D^{(t)}(\epsilon_1Q^{(t)}_++\epsilon_1Q^{(t)}_-)\right)}\right),  
\nonumber\\
\leq &&
 \Theta^{(t)}e^{-\epsilon_1\Tr{(\sigma_D^{(t)}(\rho-\sigma_G^{(t)}))}}~,
\end{eqnarray}
where the first inequality and the second inequality come from the Golden-Thompson inequality, the third inequality employs the conclusion of Corollary \ref{coro:clm7}, the penultimate inequality is supported by $-\Tr(\sigma_D^{(t)}(\epsilon_2-\epsilon_1)Q^{(t)}_-) \geq 0$ (To be proved later), and the last inequality is supported by $1-x\leq e^{-x}$. We now prove that the last second inequality always satisfies $-\Tr(\sigma_D^{(t)}(\epsilon_1-\epsilon_2) Q^{(t)}_-)\geq 0$. Since $\sigma_D^{(t)}\succeq 0$, it is equivalent to prove $(\epsilon_1-\epsilon_2) Q^{(t)}_-\preceq   0$. Due to $Q_-^{(t)}\preceq 0$, $(\epsilon_1-\epsilon_2) Q^{(t)}_-\preceq   0$  is reduced to prove $\epsilon_2-\epsilon_1\geq 0$, which is always succeed since $e^{\epsilon}+e^{-\epsilon}>2$. 

By induction, we have 
\begin{equation}\label{eqn:33}
	\Theta^{(T+1)}\leq \Theta^{(1)}e^{\sum_{t=1}^T \Tr(\sigma_D^{(t)}(\rho-\sigma_G^{(t)})) }~. 
\end{equation}
Due to $\Theta^{(1)}=\Tr(e^{-\epsilon (\rho-\sigma_G^{(1)})) })\leq \Tr(e^{\mathbb{I}})\leq e\Tr(\mathbb{I}_{AB})\leq e^{N+1}$, the above inequality can be reformulated as  
\begin{equation}
	\Theta^{(T+1)}\leq e^{N+1}e^{\sum_{t=1}^T \Tr(\sigma_D^{(t)}(\rho-\sigma_G^{(t)})) }~. 
\end{equation}

Concurrently, we have 
\begin{equation}\label{eqn:clm7_1}
	\Theta^{(T+1)} \geq e^{\sum_{t=1}^T\Tr{(\sigma_D^{*}(\rho-\sigma_G^{(t)}))}}~,
\end{equation}
where the proof is analogous to   Eqn.~(\ref{eqn:21}). We specifically  set $A=\sum_{t=1}^T(\rho-\sigma_G^{(t)})$. The left term  $\Theta^{(T)}$ of Eqn.~(\ref{eqn:clm7_1}) follows  $\Theta^{(T)}=\Tr(\exp(-\epsilon A))=\sum_k \exp{(-\epsilon\lambda_k(A))}\geq \exp{(-\epsilon\lambda_n(A))}$, where  $\lambda_n(A)$ refers to the minimum absolute eigenvalue  of $A$. The right term of Eqn.~(\ref{eqn:clm7_1})  follows $e^{-\epsilon\sum_{t=1}^T\Tr{(\sigma_D^{*}A)}}\leq \exp{(-\epsilon\lambda_n(A))}$, since we always have $\lambda_n(A)\leq \Tr(A\sigma_D^{*})$ with $0\preceq \sigma_D^{*}\preceq \mathbb{I}_{AB}$ and $\Tr{\sigma_D^{*}}=1$. This leads to $\exp{-\epsilon\lambda_n(A)}\geq \exp{-\epsilon\Tr(A \sigma_D^{*})}$. Therefore, we obtain Eqn.~(\ref{eqn:clm7_1}).

By connecting Eqn.~(\ref{eqn:33}) with Eqn.~(\ref{eqn:clm7_1}), we have 
\begin{equation}\label{eqn:clm7_2}
 e^{-\epsilon\sum_{t=1}^T\Tr{(\sigma_D^{*}(\rho-\sigma_G^{(t)}))}}
 \leq  e^{N+1}  e^{-\epsilon_1\sum_{t=1}^T\Tr{(\sigma_D^{(t)}(\rho-\sigma_G^{(t)}))}} ~.
\end{equation}

Taking logarithms and simplifying Eqn.~(\ref{eqn:clm7_2}), we obtain the  following inequality,
\begin{align*}\label{eqn:clm7_3}
&  -\epsilon \sum_{t=1}^T \Tr{(\sigma_D^*(\rho-\sigma_G^{(t)}))} \leq (N+1)-\epsilon_1 \sum_{t=1}^T\Tr{(\sigma_D^{(t)}(\rho-\sigma_G^{(t)}))}
\nonumber\\
\Rightarrow &  -\epsilon \sum_{t=1}^T\Tr{(\sigma_D^*(\rho-\sigma_G^{(t)}))} \leq (N+1)-\epsilon(1-\epsilon )\sum_{t=1}^T\Tr{(\sigma_D^{(t)}(\rho-\sigma_G^{(t)}))}
\nonumber\\
\Rightarrow &(1-\epsilon)\sum_{t=1}^T\Tr{(\sigma_D^{(t)}(\rho-\sigma_G^{(t)}))} \leq \sum_{t=1}^T \Tr{(\sigma_D^*(\rho-\sigma_G^{(t)}))}+ \frac{N+1}{\epsilon}\nonumber\\
\Rightarrow & \sum_{t=1}^T \Tr{(\sigma_D^{(t)}(\rho-\sigma_G^{(t)}))} \leq 
  \sum_{t=1}^T \Tr{(\sigma_D^*(\rho-\sigma_G^{(t)}))} + \epsilon T + \frac{N+1}{\epsilon} \nonumber\\
  \Rightarrow & \sum_{t=1}^T\mathcal{L}(\sigma_G^{(t)}, \sigma_D^{(t)}) - \frac{\epsilon^2 T + (N+1)}{2\epsilon}
 \leq \sum_{t=1}^T\mathcal{L}(\sigma_G^{(t)}, \sigma_D^*)~.
\end{align*}
The first arrow results from the fact that $\epsilon_1\geq \epsilon(1-\epsilon)$. The second arrow comes from dividing $\epsilon$ on both sides. The third arrow employs $\Tr{(\sigma_D^*(\rho-\sigma_G^{(t)}))}\leq 1$, since $\Tr{(\sigma_D^*(\rho-\sigma_G^{(t)}))}=\Tr{(\sigma_D^*\rho)}-\Tr{(\sigma_D^*\sigma_G^{(t)}))}\leq \Tr{(\sigma_D^*\rho)}\leq 1$. The last arrow comes from the definition of the loss function of QMMW. This inequality immediately indicates that the bound for the minimized regret $R(\sigma_D)$ is $R(\sigma_D)\leq \varepsilon_D =(\epsilon^2 T + N+1)/(2\epsilon T)$. 
\end{proof}

 \subsection{SM (B) Instantiation of QMMW} 
 Let us briefly review the Quantum Multiplicative Matrix Weight algorithm  (QMMW). QMMW consists of three steps: First, initializing parameters; Second,  updating the generated state  $\sigma_G^{(t)}$ and the density operator  $\sigma_{D}^{(t)}$ iteratively during $T$ training rounds; Last, calculating the loss $\mathcal{L}(\sigma_G^{(t)}, \sigma_D^{(t)})$ using $\bar{\sigma}_G$ and $\bar{\sigma}_D$ with $\bar{\sigma}_G=\sum_{t=1}^T\sigma_G^{(t)}/\sqrt{T}$ and $\bar{\sigma}_D=\sum_{t=1}^T \sigma_D^{(t)}/T$.  Both  the generated state $\sigma^{(t)}_G$ and the discriminator $\sigma_D^{(t)}$ can be treated as  Gibbs state. The formal definition of Gibbs state and Gibbs Sampler is: 
 \begin{definition}[Gibbs Sampler \cite{van2018improvements}] \label{def:gibss-sam}
 A $\theta$-precise Gibbs-sampler is a unitary that creates as output a purification of a $\theta$-approximation in trace distance of the Gibbs state  $\exp{(\sum_{t=1}^T-y_t C^{(t)})}/\Tr{(\exp(\sum_{t=1}^T-y_t C^{(t)}))}$, where $\{C^{(t)}\}_{t=1}^T$ is a set of Hermitian matrix. If $\|y\|_1\leq K$ and the support of $y$ has the size at most $d$, then we write  $\mathcal{T}(K, d, 4\theta)$ for the cost of this unitary.

 We also allow Gibbs-samplers that require a random classical input seed $S\in\{0,1\}^a$ for some $a=\mathcal{O}(\log (1/\theta))$. In this case the output should be a $\theta$-approximation of the Gibbs state with high probability $(\geq 4/5)$ over a uniformly random input seed $S$.
\end{definition}
 
For ease of description, we denote the responsible $\theta$-precise Gibbs-sampler for the generated  state $\sigma^{(t)}_G$  as  $U_{\sigma^{(t)}_G}$. For the discriminator $\sigma_D^{(t)}$, we denote the responsible $\theta$-precise Gibbs-sampler as  $U_{\sigma_D^{(t)}}$. Observing the QMMW algorithm, the third step can be efficiently executed using the SWAP test once we have prepared  $\{U_{\sigma^{(t)}_G}\}_{t=1}^T$   and $\{U_{\sigma_D^{(t)}}\}_{t=1}^T$ \cite{buhrman2001quantum}, where the query complexity is $\mathcal{O}(1)$.  Consequently,  preparing Gibbs samplers is the central part of the  implementation of QMMW and also dominates the computational cost.

We now elaborate how to accomplish the second step of QMMW, i.e.,  the construction of $\{U_{\sigma^{(t)}_G}\}_{t=1}^T$   and $\{U_{\sigma_D^{(t)}}\}_{t=1}^T$. This task employs  two subroutines  $O_{\sigma_G^T}$ and $O_{\sigma_D^T}$, i.e.,  the subroutine $O_{\sigma_G^T}$ after training $t$-rounds is 
\begingroup
\allowdisplaybreaks
\begin{align*}
	\begin{cases}
		O_{\sigma_G^T}\ket{\tau}\ket{0}^{\otimes N'} :=\sum_{\tau\leq t} \alpha_{\tau} \ket{\tau}U_{\sigma^{(\tau)}_G}\ket{0}^{\otimes N'} = \sum_{\tau\leq t} \alpha_{\tau} \ket{\tau}\ket{\psi_G^{(\tau)}}, & \tau\leq t \\
O_{\sigma_G^T}\ket{\tau}\ket{0}^{\otimes N'} :=\sum_{\tau> t} \alpha_{\tau} \ket{\tau}\ket{0}^{\otimes N'} =\sum_{\tau> t} \alpha_{\tau} \ket{\tau}\ket{0}^{\otimes N'}, & \tau> t
	\end{cases}
\end{align*}
\endgroup
where $N'=a+N$, $a$ refers to the number of ancillary qubits with $a\sim O(\log N)$ \cite{van2018improvements}, $\ket{\tau}$ refers to the computational basis corresponding to the $\tau$-th training round,  $U_{\sigma^{(\tau)}_G}$ prepares the purification $\ket{\psi_G^{(\tau)}}$  of the Gibbs state $\sigma_{G}^{(\tau)}$, and $\sum_{\tau=1}^t \alpha_{\tau}^2=1$. Note that the Gibbs sampler $U_{\sigma^{(\tau)}_G}$ can only be prepared for $\tau \leq t$. Similarly, the subroutine $O_{\sigma_D^T}$ after training $t$-rounds is 
\begingroup
\allowdisplaybreaks
\begin{align*}
	\begin{cases}
		O_{\sigma_D^T}\ket{\tau}\ket{0}^{\otimes N'} :=\sum_{\tau\leq t} \beta_{\tau}\ket{\tau}U_{\sigma^{(\tau)}_D}\ket{0}^{\otimes N'} =\sum_{\tau\leq t} \gamma_{\tau} \ket{\tau}\ket{\psi_D^{(\tau)}}, & \tau\leq t \\
O_{\sigma_D^T}\ket{\tau}\ket{0}^{\otimes N'} := \sum_{\tau> t} \gamma_{\tau} \ket{\tau}\ket{0}^{\otimes N'} = \sum_{\tau> t} \gamma_{\tau} \ket{\tau}\ket{0}^{\otimes N'}, & \tau> t
	\end{cases}
\end{align*}
\endgroup
where $\ket{\psi_D^{(\tau)}}$ refers to the purification of the Gibbs state $\sigma_{G}^{(\tau)}$. After training $T$ steps with setting $\alpha_{\tau}=\sqrt{1/T}$ and $\gamma_{\tau}=\sqrt{1/T}$ for any $\tau\in [T]$, we prepare the purification of Gibbs state $\bar{\sigma}_G^{(t)}$ and $\bar{\sigma}_D^{(t)}$. 

In QMMW, we employ the method proposed by \cite{van2018improvements} to prepare the Gibbs sampler, which has the following result,
\begin{thm}[Theorem 22, \cite{van2018improvements}]\label{thm:22}
	 Suppose that we have query access to the  unitaries $U_{\rho^{\pm}}$ preparing a purification of the subnormal density operators \footnote{A density operator $Q$ is said to be subnormal if it satisfies $\Tr{Q}=1$.} $\mu^{\pm}$, such that $H=(\mu^+ - \mu^-)/2$. Suppose that $\beta\geq 1$ and $\theta, \delta\in (0, 1]$,  there is a quantum algorithm, that using $\mathcal{O}_{\theta}(\beta^{3.5}/\delta)$ queries to controlled-$U_{\rho^{\pm}}$ or their inverses, prepares a purification of a quantum state $\rho_S$ such that $
	 	\left|\rho_S-\frac{e^{-\beta H}}{\Tr\left(e^{-\beta H} \right)}\right|\leq \theta~,$
	 where $S$ is an $\mathcal{O}(\log (\beta/\delta))$-bit random seed, and the above holds for at least $(1-\delta)$-fraction of seeds.
\end{thm}
The update rule of QMMW accompanies  with the query access to $U_{\sigma_D^{(1)}}$ (A set of  Hadamard gates to prepare the maximally mixed state $\sigma_D^{(1)}$)  and $U_{\rho}$ enables us to use Theorem \ref{thm:22} to construct all Gibbs samplers. Without loss of generality, we consider the preparation of $U_{\sigma^{(t)}_G}$ and $U_{\sigma^{(t)}_D}$. Following the update rule, the density operator $\mu^+$ and $\mu^-$ defined in Theorem \ref{thm:22} refers to  $0$ and the purification of  $\sum_{\tau=1}^{t-1} \sigma_D^{(\tau)}/(t-1)$, respectively to prepare $U_{\sigma^{(t)}_G}$. Meanwhile, we have $\beta=2\epsilon(t-1)$. It is easy to see that $\beta(\mu^+-\mu^-)/2=\sum_{\tau=1}^{t-1}-\epsilon \sigma_D^{\tau}$ is the exponential term  for updating  $\sigma_G^{(t)}$. The purified state $\mu^+$ can be generated by querying the subroutine $O_{\sigma_D^T}$ once, i.e., with setting $\gamma_{\tau}=\sqrt{1/(t-1)}$ for any $\tau\in[t-1]$. Likewise, to prepare $U_{\sigma^{(t)}_D}$, the density operator $\mu^+$ and $\mu^-$ refers to the purification of $\rho$ and $\sum_{\tau=1}^{t-1} \sigma_G^{(\tau)}/(t-1)$, respectively. Meanwhile, we have $\beta=2\epsilon(t-1)$.   The corresponding purification state of $\mu^+$ and $\mu^-$ can be prepared by querying $U_{\rho}$ and $O_{\sigma_G^T}$ once, with setting $\gamma_{\tau}=\sqrt{1/(t-1)}$ for any $\tau\in[t-1]$. By induction, we can prepare all Gibbs samplers and build two subroutines after $T$ training rounds.

\subsection{SM (C) Proof of Theorem 2.}
\begin{proof}
As discussed in the previous subsection, the main computational cost of QMMW is in the preparation of two subroutines, or equivalently a set of Gibbs samplers. We now employ the conclusion of Theorem \ref{thm:22} to characterize the computation cost of building two subroutines. Observing the conclusion of Theorem \ref{thm:22}, the computational  cost is highly related to the two variables, i.e., $\beta$ and $\delta$. Notably, the error $\theta$  in Theorem \ref{thm:22} is caused by loading classical input into quantum state, which is not required in QMMW.  This error can thus be eliminated in our case and the query complexity transformed to $\mathcal{O}(\beta^{3.5}/\delta)$. 

At $(t+1)$-th step, we have $\beta = 2 \epsilon t$ with $\epsilon=\sqrt{N/T}$, which leads to the cost $\mathcal{O}(\beta^{3.5}/\delta) \leq  \mathcal{O}((NT)^{3}/\delta)$. Since each training round requires at most $\mathcal{O}((NT)^{3}/\delta)$ query complexity to prepare two Gibbs samplers, the total query complexity for building two subroutines is $\mathcal{O}(N^{3}T^4/\delta)$ after $T$ training rounds. By setting $\delta$ as a small constant, the query complexity of QMMW is $\mathcal{O}(N^{3}T^4)$,  with  the first step and the third step  requiring only $\mathcal{O}(1)$ query complexity.

The runtime cost of QMMW is also $\mathcal{O}(N^{3}T^4)$. The first and third steps of QMMW only require $\mathcal{O}(1)$ and $\mathcal{O}(\log N)$ elementary operations. In second step, $\mu^+$ and $\mu^-$ can be prepared with $\mathcal{O}(1)$ queries and using $\mathcal{O}(poly \log(\beta))$ elementary operations \cite{van2018improvements}. Therefore, the total runtime complexity for QMMW is  $\mathcal{O}(N^{3}T^4)$.
\end{proof}

\subsection{SM(D) Proof of Lemma  3}
\begin{proof}[Proof of Lemma 3]
The zero-sum game played by QuGAL is identical to that of QuGAN, since both  possess the convex-concave property and have the same equilibrium, supported by the linear property of trace operations, and $\rho$ is sampled from a convex set. Concretely, we have  
\begin{eqnarray}
&&\max_{\sigma}\min_{\mathcal{D}}\mathcal{L}(\sigma, \mathcal{D}) \nonumber \\
=&&\max_{\sigma}\min_{\mathcal{D}} \Tr((\mathbb{I}-\mathcal{D})\sigma)P(G)+\Tr(\mathcal{D}\rho)P(R) \\
=&&\max_{U_G}\min_{U_D} \Tr((\mathbb{I}\otimes (\mathbb{I}-E_F))U_D(\sigma_G\otimes\ket{0}\bra{0})U_D^{\dagger})P(G)+\Tr((\mathbb{I}\otimes E_F)U_D(\rho\otimes\ket{0}\bra{0})U_D^{\dagger})P(R)\nonumber~, 
\end{eqnarray}  
where $\sigma_G =\Tr_a(U_G(\ket{0}\bra{0})^{\otimes N'}U_G^{\dagger} )$. The first equality  follows the definition of loss function of  QuGAL, which is  
$$\mathcal{L}(\sigma, \mathcal{D}) = \Tr((\mathbb{I}-\mathcal{D})\sigma)P(G)+\Tr(\mathcal{D}\rho)P(R) ~.$$ The second equality employs Naimark’s dilation theorem and Stinespring's dilation theorem. 

 \end{proof}
 
 \subsection{SM (E) Generative Adversarial Network}
The generator $G$ and discriminator $D$ for classical GANs  are typically implemented by multi-layer neural networks \cite{goodfellow2016nips}. The generator can be treated as a function $G$, which aims to map from a random variable $\bm{z}\in\mathbb{R}^{|\bm{z}|}$ sampling from the latent space to the data space $\bm{x}\in\mathbb{R}^{|\bm{x}|}$. Mathematically, we have $G:G(\bm{z})\rightarrow \bm{x}$, where $|\cdot|$ denotes the number of dimensions with $|\bm{z}|\ll |\bm{x}|$, and $\bm{x}$ refers to the generated data (e.g., an image with $|\bm{x}|$ pixels). Discriminator $D$ may be similarly characterized as a function that maps from input data to the class distribution: $D : D(\bm{x})\rightarrow (0,1) $, where the training data are expected to be $1$ (True) and the generated data are expected to be $0$ (False). If the distribution learned by the generator is able to match the real data distribution perfectly,  the discriminator will be maximally confused, predicting $0.5$ for all inputs. This unique solution whereby $D$ can never discriminate between the generated data and the training  data is called Nash equilibrium \cite{goodfellow2014generative}.

The training of GANs involves  finding the parameters of a discriminator $D$ to maximize  classification accuracy, and finding the parameters of a generator $G$ to maximally confuse the discriminator. The performance of GAN is evaluated using a loss function $L(G,D)$ which  depends on both the generator and the discriminator. The training procedure can be treated as:
\begin{equation}
\min_G\max_D L(G,D),
\end{equation}
where $	L(G,D) = \mathbb{E}_{\bm{x}\sim p_{data}(\bm{x})} [\log D(\bm{x})] +\mathbb{E}_{\bm{z}\sim p(\bm{z})} [\log(1-D(G(\bm{z})))]$, $p_{data}(\bm{x})$ refers to the distribution of  the training dataset, and $p(\bm{z})$ is the probability distribution of the latent variable $\bm{z}$. During training, the parameters of one model are updated, while the parameters of the other are fixed. 

			\end{document}